\documentclass[12pt]{article}
\usepackage[numbers]{natbib}
\usepackage{amsmath} 
\usepackage{amssymb}
\usepackage{amsthm}
\usepackage{graphicx}
\usepackage{fancyhdr}
\usepackage{array}
\newcolumntype{M}{>{\centering\arraybackslash}m{\dimexpr.25\linewidth-2\tabcolsep}}

\begin{document}


\title{\Large
Bounds on the Size of Sound Monotone Switching\\Networks Accepting Permutation
Sets of Directed Trees
}

\author{\normalsize
Joshua Brakensiek and Aaron Potechin$^{a}$\\
{\normalsize ${}^a$Massachusetts Institute of Technology}
}

\date{} 

\maketitle

\begin{abstract}

In this paper, we prove almost tight bounds on the size of sound monotone switching networks accepting permutations sets of directed trees.  This roughly corresponds to proving almost tight bounds bounds on the monotone memory efficiency of the directed ST-connectivity problem for the special case in which the input graph is guaranteed to have no path from $s$ to $t$ or be isomorphic to a specific directed tree.
\end{abstract}


\newtheorem{thm}{Theorem}[section]
\newtheorem{prop}[thm]{Proposition}
\newtheorem{lem}[thm]{Lemma}
\newtheorem{cor}[thm]{Corollary}
\newtheorem*{sublem}{Sublemma}
 
\theoremstyle{definition}
\newtheorem{df}{Definition}[section]
\newtheorem{conj}{Conjecture}[section]

\theoremstyle{remark}
\newtheorem*{note}{Note}
\newtheorem*{rem}{Remark}
\newtheorem*{case}{Case}
\newtheorem*{subcase}{Subcase}
\section{Introduction}

One long-standing open problem in computational complexity theory is the
minimal space complexity of the ST-connectivity problem.   The
ST-connectivity problem is formulated as follows: given a directed graph $G$ with starting and ending vertices $s$ and $t$, is there a directed path from $s$
to $t$?  This inquiry is not difficult to answer; the challenge
is to answer this question with the minimal amount of space
necessary.  In a celebrated result by Savitch \cite{Savitch1}, it was
shown that one can answer this problem on an $n$-vertex graph with
$O((\log n)^2)$ space.  Reingold
\cite{Reingold1} showed that if $G$ is an undirected graph, then only
$O(\log n)$ space is required.

One type of computation which gives insight into this problem is
monotone computation.  This kind of computation operates by making deductions from the existence of edges;
it does not make deductions from the absence of certain edges.  We
analyze the ST-connectivity problem using a structure called a \emph{monotone switching network}.  Defined more precisely
in Section \ref{subsec:def}, a monotone switching network is an undirected graph
with labeled edges based on queries that a program may make about the
existence of edges in the input graph.  We can think of the vertices of this network as
representing possible memory states of a program.

Finding bounds on the size of the monotone switching network roughly corresponds to finding 
bounds on the amount of space needed to compute ST-connectivity in a
monotone computation model.  Potechin~\cite{Potechin1} has shown
that in the general case, a monotone switching network needs a size of
$n^{\Theta(\log n)}$, which corresponds to needing $\Theta((\log n)^2)$
space.  However, finding lower bounds on the size of monotone
switching networks does not give us lower bounds on the amount of space
needed but does tell us the limits of monotone computation.
 To obtain general lower bounds on the amount of memory
needed, one must analyze a broader  class of switching networks, the
non-monotone switching networks~\cite{Potechin1}.

We determine bounds on the sizes of monotone switching networks
for special cases of the ST-connectivity problem.  In these special
cases, we assume that the input graph is isomorphic to a given graph via
permutation of the vertices.  For example the results in Theorem \ref{thm:out_tree} concern the case where
every vertex in the given graph has a unique path from $s$ to itself.

In our main result, Theorem \ref{thm:all_tree}, bounds are found in the case of a general
tree.  If we define $m(\sigma(G))$ to be the size of the monotone
switching network and $\ell$ to be the length of the path from $s$ to $t$, we found upper and lower bounds $B_1$ and $B_2$, respectively, on $m(\sigma(G))$ such that $$\log (B_1/B_2)\le O(\log \log \ell).$$  The previous best bounds satisfied $\log (B_1/B_2)\le O(\log \ell)$~\cite{Potechin1}.

\subsection{Outline}

In Section \ref{subsec:def}, we formally define monotone switching
networks and related terminology.  In Section \ref{sec:prev_results}, we
summarize previous work with monotone switching networks.  In Section \ref{subsec:tech}, we present techniques for bounding the sizes of certain classes of monotone switching networks which are
crucial in obtaining the results in this paper.   Section \ref{sec:main_results} provides proof of the main
result, bounding the size of sound monotone switching networks in the
case of general directed trees.

\section{Preliminary Definitions}\label{subsec:def}

To discuss monotone switching networks and their properties,
the following terminology was introduced by Potechin~\cite{Potechin1},
which we also use.

\begin{df}
Given a set of vertices $V\cup \{s,t\}$, define a \textbf{\emph{monotone switching
network for directed connectivity}} as an undirected
graph $G'$ on the set of vertices $V'\cup \{s',t'\}$.  Each edge between two vertices of $G'$ is given
a label of the form $a\to b$ where $a,b\in V\cup \{s,t\}$.
\end{df}

\begin{note}
For succinctness, we refer to monotone switching networks for directed connectivity as \emph{monotone switching networks}. 
\end{note}

An example of a monotone switching network is depicted in Figure \ref{fig:msn}.

\begin{figure}
\begin{center}
\includegraphics[height=2.2in, bb=0 0 120 160]{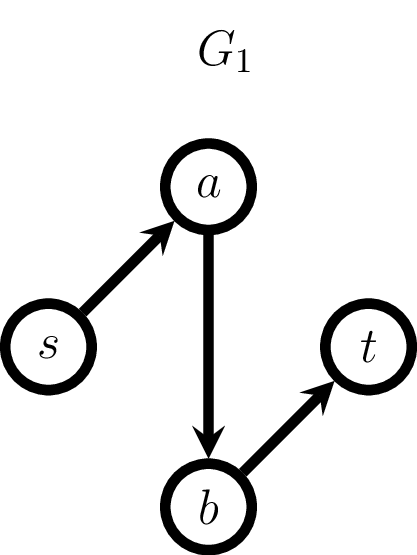} 
\hfill
\includegraphics[height=2.2in, bb=0 0 120 160]{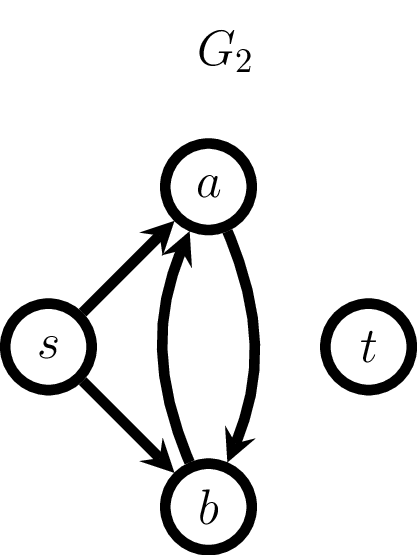}
\hfill
\includegraphics[height=2.2in, bb=0 0 212 252]{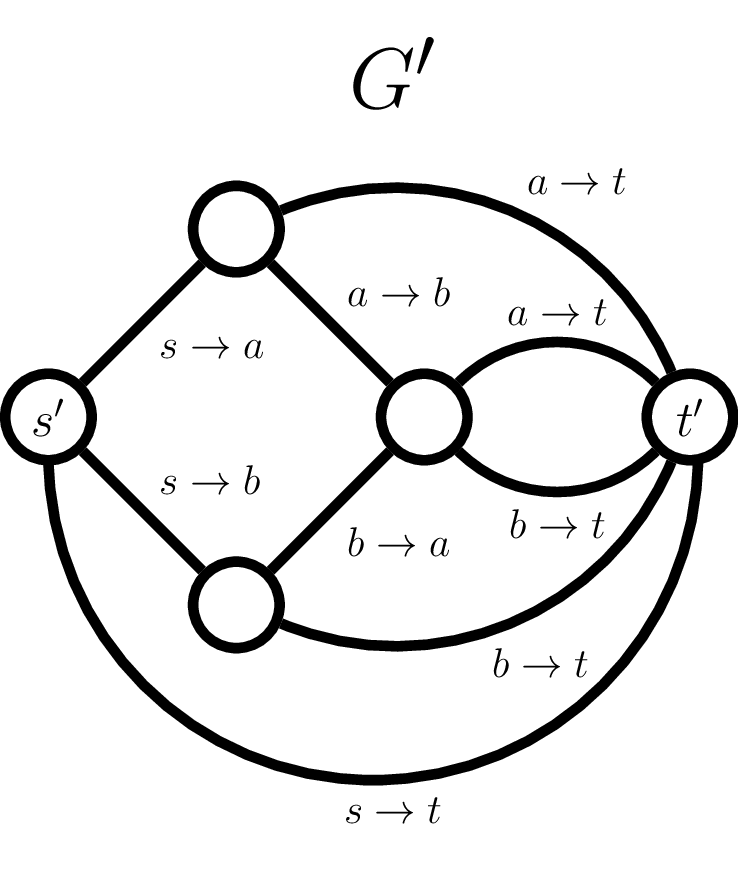}
\end{center}
\caption{Two input graphs $G_1$ and $G_2$  and a monotone switching
  network $G'$.  Notice that the monotone switching network accepts
  $G_1$ because $G'$ has a path from $s'$ to $t'$ with labels $s\to a$, $a\to b$, and
  $b\to t$, each of which is in $G_1$. Conversely, $G'$ rejects
  $G_2$ because none of the edges to $t'$ have labels which are in
  $G_2$.  More generally, $G'$ is both complete and sound.} 
\label{fig:msn}
\end{figure}

\begin{df}
Define the \textbf{\emph{size}} of a monotone switching network $G'$ as the number of
vertices of $G'$.
\end{df}

We wish to analyze how $G'$ relates to various graphs $G$ on the set of
vertices $V\cup \{s,t\}$.  Definition \ref{df:accept} quantifies this.

\begin{df}\label{df:accept}
  Given a  directed graph $G$ on $V\cup \{s,t\}$, called the \textbf{\emph{input
    graph}}, say that a monotone
switching network $G'$ \textbf{\emph{accepts}} $G$ if and only if there exists a path
from $s'$ to $t'$ in $G'$ such that the label of each edge of the path
corresponds to an edge of $G$.  For example, the label $a\to b$
corresponds to the directed edge from $a$ to $b$ in $G$.  If $G'$ does not accept $G$, then $G'$
\textbf{\emph{rejects}} $G$.   
\end{df}
For an example, see Figure \ref{fig:msn}. We analyze monotone switching networks based on which graphs
they accept and reject.

\begin{df}\label{df:complete}
A monotone switching network $G'$ is \textbf{\emph{complete}} if it
accepts any input graph $G$ for which there is a path from $s$ to $t$.
\end{df}

\begin{df}\label{df:sound}
A monotone switching network $G'$ is \textbf{\emph{sound}} if it rejects any input graph
$G$ for which there is no path from $s$ to $t$.
\end{df}
Unless explicitly stated, we assume that all monotone switching networks
under consideration are sound, which means that the
computations which the monotone switching network simulates involve sound
logical reasoning. On the other hand, almost none of the monotone
switching networks under consideration are complete.  In other words, the monotone switching networks may find the existence of a path from $s$ to $t$ for some input graphs but not others.

\begin{df}
Given a set $I$ of input graphs on $V\cup \{s,t\}$, where for each graph
$G\in I$ there is a path from $s$ to $t$, define $m(I)$ to be
the smallest possible size of a sound monotone switching network which
accepts all the elements of $I$.
\end{df}    

\begin{figure}
\begin{center}
\includegraphics[width=1.6in, bb=0 0 120 157]{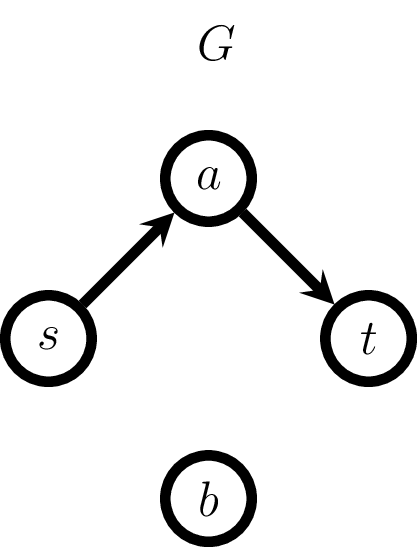}
\hfill 
\includegraphics[width=1.6in, bb=0 0 120 157]{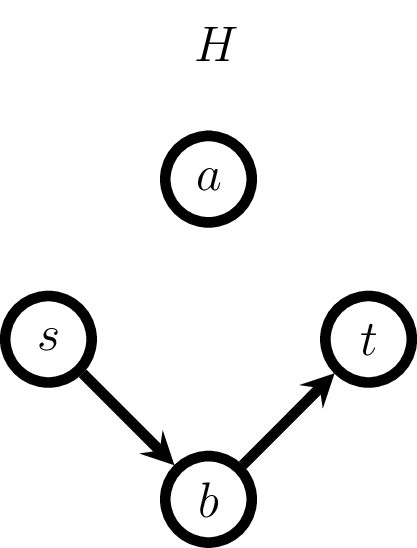}
\hfill
\includegraphics[width=1.6in, bb=0 0 120 158]{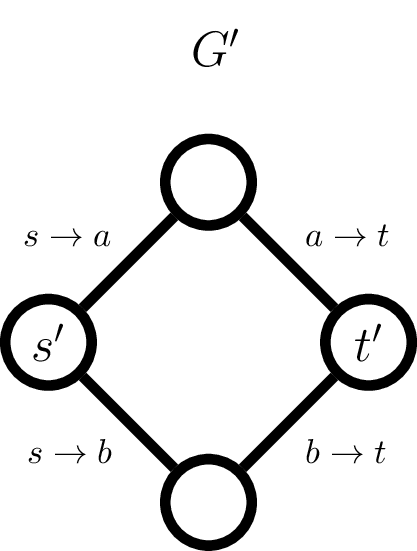}
\end{center}
\caption{By definition of $\sigma$, we have that $\sigma(G)=\{G,H\}$.  Notice that the monotone switching
  network $G'$ accepts both graphs and is sound.  Thus, $m(\sigma(G))\le 4$.
In fact, $m(\sigma(G))=4$.}
\label{fig:sigma}
\end{figure}

In Sections \ref{sec:prev_results} and \ref{sec:main_results}, we find bounds on the value of $m(I)$ for specific sets of graphs
$I$.  The sets of graphs we primarily investigate are permutation sets.

\begin{df}
Let $G$ be a directed graph on the set of vertices $V\cup \{s,t\}$.  For
any subset $W$ of $V\cup \{s,t\}$, define $\sigma_{W}(G)$ to be the set
of graphs which are all possible permutations of the labels of the vertices $V\cup
\{s,t\}$ that fix all vertices in $W$.  Let $\sigma(G)=\sigma_{\{s,t\}}(G)$.  These sets
are called \textbf{\emph{permutation sets}}.
\end{df}
 
We bound the value of $m(\sigma(G))$ for various graphs
$G$.  An example is given in Figure \ref{fig:sigma}.  To aid in finding these bounds, the
results listed in Section \ref{sec:prev_results} are used.

\section{Previous Results} \label{sec:prev_results}

The results discussed in this section, concerning the value of $m(I)$ for various sets of
input graphs $I$, were discovered by Potechin \cite{Potechin1, Potechin2}. We assume that for each $G\in I$, its vertices are taken from the set of
$n$ vertices $V\cup \{s,t\}$.  Let $\mathcal P$ be the set of directed
graphs with $n$ vertices such that there is a path from $s$ to $t$.  We then have the
following theorem about $\mathcal P$.

\begin{thm}[Potechin~\cite{Potechin1}]\label{thm:all}
We have that $$m(\mathcal P)=n^{\Theta(\lg n)},$$ where $\lg n$ stands
for $\log_2(n)$.
\end{thm}

\begin{note} The bound we get for $m(\mathcal P)$ uses big $\Theta$
notation in the \emph{exponent}, instead of as a constant factor as these are the best bounds currently known. These
bounds are tight enough for our purposes because they
heuristically correspond to an algorithm using $O(\lg (m(\mathcal
P)))=O((\lg n)^2)$ memory, which is accurate to a
constant factor.
\end{note}

Let $\ell$ be a positive integer less than $n$.  Consider $\mathcal P_\ell$, the set of directed graphs such that there
is a path from $s$ to $t$ with length $\ell$.  Length is defined to be
the number of  edges along the path. Theorem \ref{thm:all_len} gives a bound for this subset of
$\mathcal P$.

\begin{thm}[Potechin~\cite{Potechin1}]\label{thm:all_len}
$$m(\mathcal P_{\ell})=n^{\Theta(\lg \ell)}.$$
\end{thm}

Notice the similarity between Theorem \ref{thm:all_len} and Theorem \ref{thm:many_paths}.

\begin{thm}[Potechin~\cite{Potechin}]\label{thm:many_paths}
Let $G$ be a graph such that every path from $s$ to $t$ is of length
$\ell$, and every vertex besides $s$ and $t$ is on exactly one such
path.  Then $$m(\sigma(G))=n^{\Theta(\lg \ell)}.$$ 
\end{thm}

\begin{figure}
\begin{center}
\includegraphics[width=3in, bb=0 0 334 238]{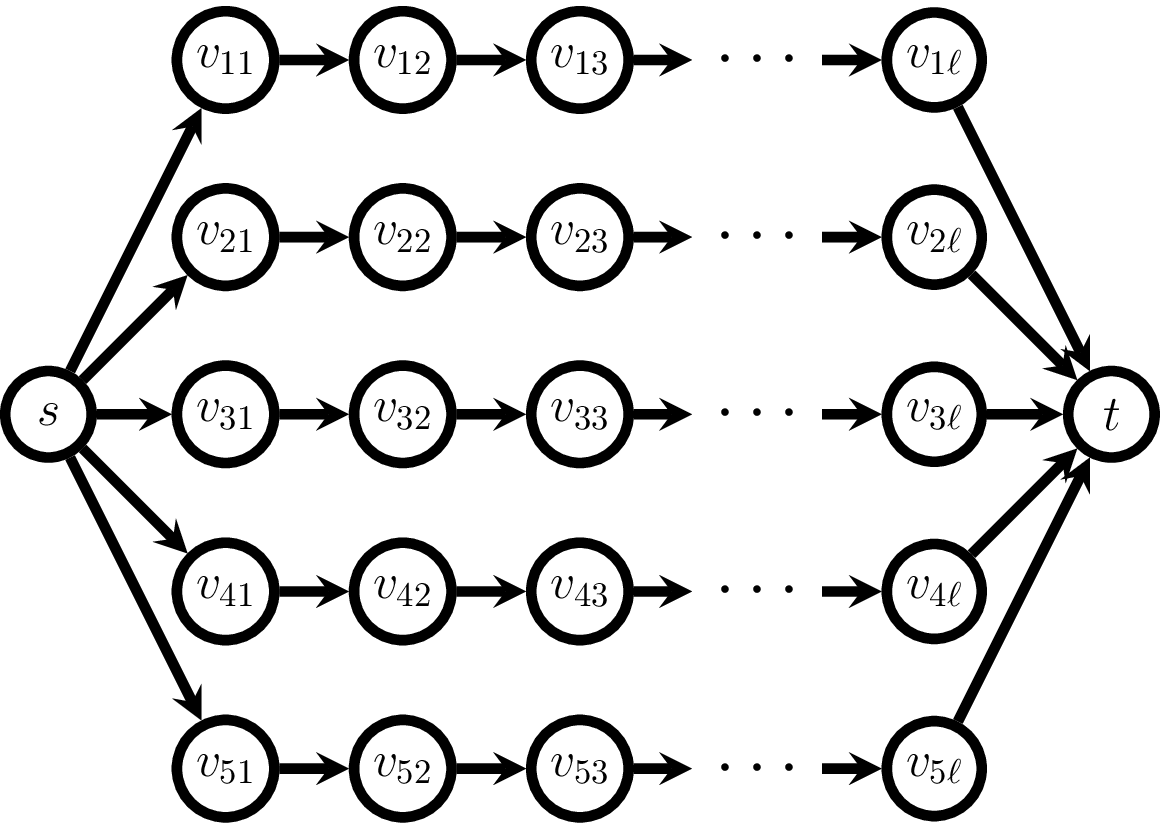}
\caption{Example input graph $G$ for Theorem \ref{thm:many_paths}.} 
\label{fig:many_paths}
\end{center}
\end{figure}

Figure \ref{fig:many_paths} has an example of $G$. The asymptotic results of Theorem \ref{thm:all_len} and Theorem
\ref{thm:many_paths} are identical, although $\sigma(G)$ is only a
small subset of $\mathcal P_{\ell}$.  In some sense, much of the
work done by the monotone switching network to accept the elements of
$\mathcal P_\ell$ is used to accept the elements of $m(\sigma(G))$.  In
contrast, Theorem \ref{thm:bbox} shows that 
some subsets of $\mathcal P_{\ell}$ can be accepted by much smaller
sound monotone switching networks.  First we need the following definition.

\begin{df}\label{def:lolli}
A vertex $v$ of a directed graph $G$ is a \textbf{\emph{lollipop}} if $s \to v$ or $v \to t$ is an edge of $G$. 
\end{df}

Theorem \ref{thm:bbox} tells us that lollipops hardly increase the asymptotic value of $m(\sigma(G))$.

\begin{thm}[Potechin~\cite{Potechin3, Potechin2}]\label{thm:bbox}
For any $n$, $k$, and $\ell$, there is a sound monotone switching network of size at most
$$n^{O(1)}k^{O(\lg \ell)}$$
which accepts all input graphs $G$ such that $G$ has $n$ vertices, there is a path of length $\ell$ from $s$ to $t$, and at most $k$ vertices are not lollipops.
\end{thm}

\begin{cor}[Potechin~\cite{Potechin3, Potechin2}]\label{cor:bbox}
Let $G$ be a graph with $n$ vertices for which there is a path of length $\ell$ from $s$ to $t$, and all but $k$ of the vertices are lollipops.  Then,
$$m(\sigma(G))= n^{O(1)}k^{O(\lg \ell)}.$$
\end{cor}

\begin{figure}
\begin{center}
\includegraphics[width=3in, bb=0 0 334 235]{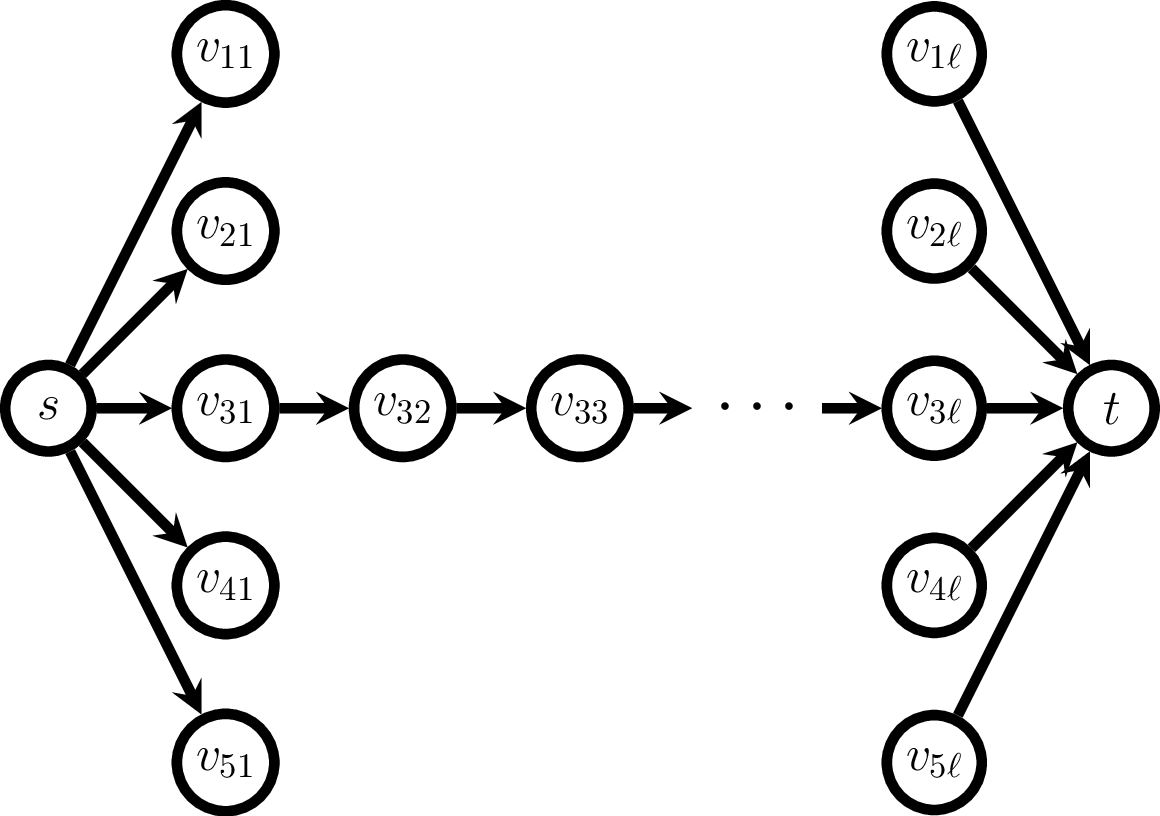}
\caption{Example input graph $G$ for Corollary \ref{cor:bbox} with many lollipops.} 
\label{fig:min_dif}
\end{center}
\end{figure}

Figure \ref{fig:min_dif} depicts an example of $G$ with a single path
from $s$ to $t$ in which all other vertices are lollipops.  When $k$ is asymptotically smaller than $n$, the value of $m(\sigma(G))$ is asymptotically much smaller for the graphs described in Corollary \ref{cor:bbox} than the graphs described in Theorem \ref{thm:many_paths}, which means that it is much easier to find the existence of a path from $s$ to $t$ for the graphs described in Corollary \ref{cor:bbox} when using monotone computation.  We use Theorem \ref{thm:bbox} when
constructing upper bounds in the proofs of Theorems \ref{thm:out_tree} and \ref{thm:all_tree}.

\section{Techniques for Bounding $m(\sigma(G))$}\label{subsec:tech}

In this section, we provide results which aid in the process of determining $m(\sigma(G))$ for arbitrary graphs $G$.

\begin{prop} \label{thm:add1}
Consider two graphs $G$ and $H$ such that every edge of $G$ is also an
edge of $H$.  If a monotone switching network $G'$ accepts $G$, then
$G'$ also accepts $H$.
\end{prop}

\begin{cor} \label{thm:add2}
Given a directed graph $G$, consider the directed graph $H$ which
results from adding an edge to $G$.  Then, $m(\sigma(G))\ge
m(\sigma(H))$.
\end{cor}

\begin{proof}
Consider a minimal-size sound monotone switching network $G'$ which accepts all
elements of $\sigma(G)$. For any element of $\sigma(G)$ there is a
corresponding element of $\sigma(H)$ with the same edges.  Hence, by
Proposition \ref{thm:add1}, $G'$ also accepts the elements of $\sigma(H)$.
Therefore, $m(\sigma(G)) \ge m(\sigma(H))$. 
\end{proof}

From Proposition \ref{thm:add1} and its corollary, we infer that for graphs
with the same number of vertices the ones with more edges typically have smaller monotone switching networks.

\begin{thm} \label{thm:repl}
Given a directed graph $G$ with an edge $a\to b$, let $H$ be the graph
where this edge is replaced with $s\to b$ and $\bar{H}$ be the
graph where this edge is replaced with $a\to t$.  Then

$$m(\sigma(G))\ge \max(m(\sigma(H)),~m(\sigma(\bar{H}))).$$
\end{thm}

\begin{proof}
We prove that $m(\sigma(G))\ge m(\sigma(H))$.  The inequality
$m(\sigma(G))\ge m(\sigma(\bar H))$ is a symmetric argument by reversing
every edge and swapping $s$ and $t$.  Consider the minimal-size  sound
monotone switching network $G'$ which accepts the elements of $\sigma(G)$.  Construct a
new monotone switching network $H'$ which contains the same vertices and
edges as $G'$.  For any edge $e$ with a label of the form $v_1\to v_2$ in $H'$, add an
additional edge, parallel to $e$, with the label $s\to v_2$.  We
now prove two properties about $H'$.
\begin{lem}
Every element of $\sigma(H)$ is accepted by $H'$.
\end{lem}
\begin{proof}
Let $H_1$ be an element of $\sigma(H)$.  Let $G_1$ be a corresponding
element of $\sigma(G)$.  Because $G_1$ is accepted by $G'$, there exists
a path $P'$ from $s'$ to $t'$ using only edges of $G_1$.  Since no
edges were deleted in the construction of $H'$, this same path $P'$ exists in
$H'$. 

If $P'$ uses only edges in $H_1$, we are done.  If not, $P'$ uses precisely one edge label which is not in $H_1$, the edge $v_1\to v_2$ which was replaced by $s\to v_2$.  Consider the path in $H'$ which follows $P'$ but instead of using the edges labeled $v_1\to v_2$, it uses the parallel edges labeled $s \to v_2$.
This is clearly an accepting path for $H_1$ in $H'$.  Thus, all the elements of $\sigma(H)$ are accepted by $H'$.
\end{proof}

\begin{lem}
$H'$ is sound.
\end{lem}
\begin{proof}
It is sufficient to prove that modifying a switching network $G'$ by adding one parallel 
edge with label $s \to v_2$ to an edge with label $v_1 \to v_2$ must preserve the 
soundness of $G'$.  Let $G'_2$ be the modified switching network and assume for sake of 
contradiction that $G'_2$ is not sound.  Then there exists a graph $G$ with no path from 
$s$ to $t$ which is accepted by $G'_2$.  This implies there is a path $P'$ from $s'$ to $t'$ 
in $H'$ which uses only edge labels in $G$. However, because $G'$ is sound, it rejects
$G$ so $P'$ must go through the one additional edge in $G'_2$, the edge labeled 
$s \to v_2$, and this edge must be in $G$. But then if we add the edge $v_1 \to v_2$ to $G$, 
we obtain a graph $G_2$ which is accepted by $G'$, as we can follow $P'$ except that we use the 
original edge labeled $v_1 \to v_2$ rather than the added parallel edge. Thus, $G_2$ must have a path 
from $s$ to $t$. But this is impossible, as if we let $V$ be the set of vertices reachable from 
$s$ in $G$, $V$ is also the set of vertices reachable from $s$ in $G_2$. To see this, 
note that $v_2$ is reachable from $s$ in $G$, so adding the edge $v_1 \to v_2$ cannot 
possibly allow us to reach any additional vertices from $s$. This is a clear contradiction, so $H'$ is sound.
\end{proof}
Since there exists a sound monotone switching network $H'$ of
size $m(\sigma(G))$ which accepts every element of $\sigma(H)$, we have that
$$m(\sigma(G))\ge m(\sigma(H))$$
as desired.
\end{proof}

Heuristically, Theorem \ref{thm:repl} implies that when the number of
vertices and edges is the same for two graphs, the one with more edges connected from
$s$ or to $t$ typically has a smaller monotone switching network.

In contrast to Theorem \ref{thm:repl}, Proposition \ref{thm:useless} demonstrates the case in which there are similar edges but in the reverse direction.

\begin{df}
An edge is \textbf{\emph{useless}} if it is of the form $v\to s$ or $t\to
v$ for some vertex $v$.
\end{df}

This definition is motivated by the fact that having an edge of this
form gives no information about whether there is a path from $s$ to $t$.

\begin{thm}\label{thm:useless}
Let $G$ be a graph with useless edges.  Let $H$ be a copy of $G$ with the useless edges removed.  Then $m(\sigma(G))=m(\sigma(H))$.
\end{thm}

\begin{proof} [Proof of Theorem \ref{thm:useless}.]
This result follows from Lemma \ref{lem1}.  The case with an edge of the form $t \to a$ follows by a symmetrical argument.

\begin{lem}\label{lem1}
Let $G$ be a graph with the edge $a\to s$.  Consider $H$, an identical
graph except $a\to s$ is removed.  Then $m(\sigma(G))=m(\sigma(H))$.   
\end{lem}

\begin{proof}
Since an edge was removed from $G$ to yield $H$, from Theorem
\ref{thm:add2}, we have that $m(\sigma(G))\le m(\sigma(H))$. Consider a sound monotone switching network $G'$ of minimal size which
accepts all the elements of $\sigma(G)$.  Replace every edge whose label is of the form
$v\to s$, for some $v$, with an unlabeled edge (an unlabeled edge can
be traversed under any condition) to produce a monotone switching network $G'_2$.  This monotone switching
network $G'_2$ accepts all the elements of $\sigma(H)$.  It is sufficient to prove now that $G'_2$ is sound.  If $G'_2$ were not sound, then a disconnected graph $K$ would exist which $G'_2$ accepts.  Thus, an accepting path $P'$ must traverse an unlabeled edge.  Let $\bar K$ be identical to $K$ except the edge $v\to s$ is added for all $v$.  We have that $\bar K$ is accepted by $G'$ by following the path $P'$, except that the unlabeled edges are replaced with edges whose labels are of the form $v\to s$.  As $G'$ is sound, $\bar K$ must a path from $s$ to $t$.  However, the addition of edges of the form $v\to s$ to a graph without a path from $s$ to $t$ cannot produce a graph with a path from $s$ to $t$. This is a contradiction; thus, $G'_2$ is sound.
\end{proof}
Therefore, $$m(\sigma(G)) = |V(H)| \ge m(\sigma(H)).$$
\end{proof}

Theorem \ref{thm:useless} is very useful for proving lower
bounds of $m(\sigma(G))$, especially when combined with Theorem
\ref{thm:merge}.  To introduce this theorem, we first define the
concept of a \emph{merge graph}.  Figure \ref{fig:merge} depicts an example of a
merge graph.

\begin{figure}
\begin{center}
\includegraphics[width=3in, bb=0 0 315 162]{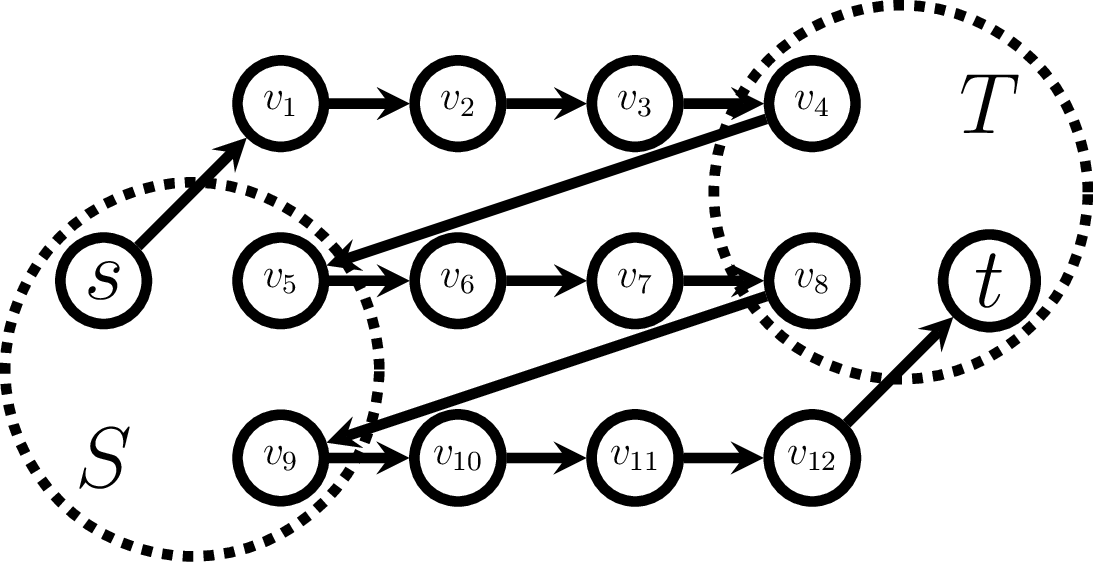}
\hfill
\includegraphics[width=3in, bb=0 0 287 130]{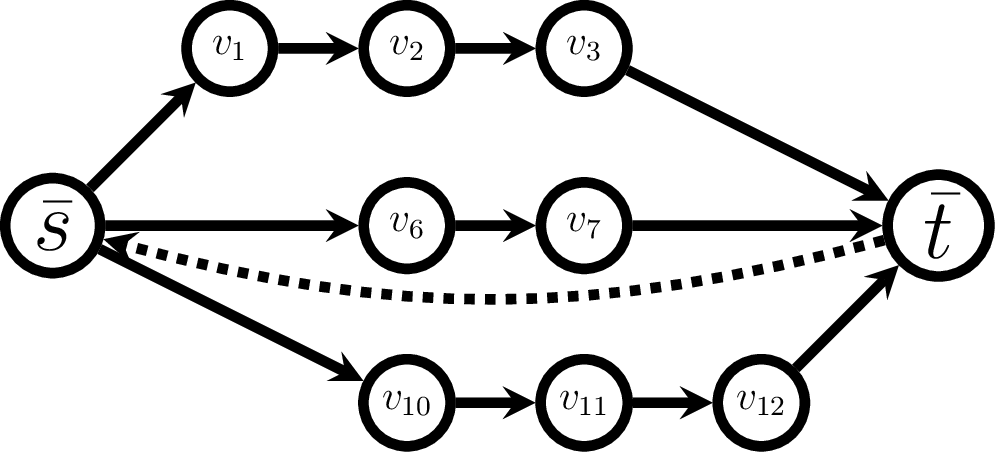}
\end{center}
\caption{Example showing $G$ (left) and the merge graph $G_{(S,T)}$ (right).  The dashed edge from $\bar t$ to $\bar s$ is useless and can be removed without affecting the value of $m(\sigma(G_{(S,T)}))$.}
\label{fig:merge}
\end{figure}

\begin{df}\label{def:merge}
Given a graph $G$, let $S$ be a set of vertices such that $s\in S$ and $t\not\in S$.
Also let $T$ be a set of vertices such that $t\in T$ and $s\not\in T$.
Consider the graph $G_{(S,T)}$ whose vertex set is identical to $G$
except the vertices of $S$ have been merged into a single vertex $\bar s$, and the vertices of
$T$ have been merged into a single vertex $\bar t$.  Any edge between
two elements of $S$ or two elements of $T$ is removed.   Any edge with
exactly one endpoint in $S$ is replaced with an edge whose corresponding
endpoint is $\bar s$.  Likewise, any edge with exactly
one endpoint in $T$ is replaced with an edge whose corresponding endpoint
is $\bar t$.  Any remaining edges remain unchanged.  Define $G_{(S,T)}$ to be the $(S,T)$\textbf{\emph{-merge graph}} of
$G$. 
\end{df} 

\begin{thm}\label{thm:merge}
Given a graph $G$, let $S$ and $T$ be subsets of $G$ defined as in Definition \ref{def:merge}.
Then $m(\sigma(G))\ge m(\sigma(G_{(S,T)}))$.    
\end{thm}

\begin{proof}
Because $\sigma_{S\cup T}(G)\subset \sigma(G)$, we have that
$m(\sigma_{S\cup T}(G))\le m(\sigma(G))$. Consider the sound monotone
switching network $H'$ which accepts all of the elements of $\sigma_{S\cup
  T}(G)$.  Construct a new monotone switching network $\bar H'$ by
taking every edge label of this monotone
switching network which has an endpoint in $S$ or $T$ and replace
it with $\bar s$ or $\bar t$, respectively.   This new monotone switching network accepts all elements of
$\sigma(G_{(S,T)})$.  As any graph without a path from $\bar s$ to $\bar
t$ cannot correspond to a graph with a path from $s$ to $t$, we have that $\bar H'$
is sound.  Thus, 
$$m(\sigma(G_{(S,T)}))\le m(\sigma_{S\cup T}(G))\le m(\sigma(G)).$$
\end{proof}

Theorem \ref{thm:useless} and Theorem \ref{thm:merge} can work together well to obtain bounds for arbitrary graphs, as shown in Section \ref{sec:main_results}.

A natural question after observing Theorem \ref{thm:merge} is whether
merging arbitrary sets of vertices which do not necessarily contain $s$
or $t$ provides the same result.  This is not
the case: there exists a graph $G$ for which contracting particular sets of vertices
results in an increase in the value of $m(\sigma(G))$.  For proof, see Appendix \ref{app:merge}.

\section{Main Results}\label{sec:main_results}

We now demonstrate bounds on the value of $m(\sigma(G))$ where $G$ is any
tree.  The result we prove is as follows.

\begin{thm}\label{thm:all_tree}
Let $G$ be an arbitrary directed tree with a path from $s$ to $t$.  Define $d_i^s$ to be the number of
vertices which are accessible from $s$ with $i$ as the maximum distance of its
descendants from $s$.  Let $d_i^t$ to be the number of vertices
which can access $t$ with the maximum distance of its ancestors
to $t$ being $i$.  Define the sequence $c^s_1,\hdots, c^s_{\lceil \lg
  \lg \ell\rceil}$ such that
\begin{align*}
c_1^s&=\sum_{i=1}^{n}d_i^s&\text{ and }&&c_k^s&=\sum_{i=2^{2^{k}}}^nd_i^s\text{ where }k\ge 2.
\end{align*}
Define $c^t_1,\hdots, c^t_{\lceil \lg \lg \ell \rceil}$ similarly.  Let $\bar d$ be the number of vertices which are not accessible from $s$ or $t$.  Let $\ell$ be the length of the path from $s$ to $t$.  Then $m(\sigma(G))$ can be bounded by
$$(\ell+\bar d)^{\Omega(\lg \ell)}\max_{1\le i\le \lceil \lg \lg
  \ell\rceil}(c_i^s+c_i^t)^{\Omega(2^i)}\le m(\sigma(G))\le n^{O(\lg \lg
\ell)}(\ell+\bar
d)^{O(\lg \ell)}\prod_{i=1}^{ \lceil \lg \lg
  \ell\rceil}(c_i^s+c_i^t)^{O(2^i)}.$$
\end{thm}

\begin{note}
If we let $B_1$ be the upper bound and $B_2$ be the lower bound, then
$$\lg (B_1/B_2)= O(\lg \lg \ell).$$
\end{note}

\noindent This proof is divided into proving the lower and upper bounds.

\subsection{Lower Bound}

First we prove the following lemma.

\begin{lem}\label{lem:float}
Let $G$ be a graph and $H$ be a tree disconnected from $G$.  Let $n_H$
be the number of vertices of $H$.  Let $P$ be
a path with $\left\lceil\sqrt{n_H}\right\rceil$ vertices which is disconnected
from $G$.  Then
$m(\sigma(G\cup H))\ge m(\sigma(G\cup P))$.
\end{lem}

\begin{rem}
The length of $P$ can be significantly increased. See Appendix \ref{app:float}.
\end{rem}

\begin{proof}
Give the vertices of $H$ a depth labeling $d(v)$ in
the following way: pick an arbitrary vertex
$v$ and let $d(v)$ be $0$.  We can now define $d$ recursively.  For any
edge $v_1\to v_2$, we have $d(v_2)-d(v_1)=1$.  If there is a path from
vertex $w_1$ to vertex $w_2$, then $d(w_2)-d(w_1)>0$.  Thus, any two
vertices of the same depth do not have a directed path between them.
Let $d_{\min}$ be the minimum depth and $d_{\max}$ be the maximum
depth. Note that $d_{\min}$ may be negative if there are edges directed toward $v$.  We now have two cases to consider.

\emph{Case 1:} $d_{\max}-d_{\min}+1< \left\lceil \sqrt{n_H}\right\rceil$.

By the pigeonhole principle, there exists a depth $\bar d$ which
has at least $\left\lceil \sqrt{n_H}\right\rceil$ vertices.  Call this set of vertices $\bar{D}$.  We can merge all of
the vertices of depth less than $\bar d$ with $t$ and all the vertices
of depth greater than $\bar d$ with $s$.  After removing useless edges, the vertices of $\bar D$ are
isolated.  We then add edges between the vertices of $\bar D$ to create a path
of length $\left\lceil \sqrt{n_H}\right\rceil$, as desired.  If there are
additional vertices, they can be merged with $s$.

\emph{Case 2:} $d_{\max}-d_{\min}+1\ge \left\lceil
  \sqrt{n_H}\right\rceil$.

First consider the case where if we ignore directions, then $H$ consists of only a single undirected path $P$
from a vertex $w_{min}$ at depth $d_{\min}$ to a vertex $w_{\max}$ at depth $d_{\max}$. Now looking at the edge directions, let $c^+$ be the number of 
edges along the path which go from a vertex of lesser depth to one of greater depth. Define $c^-$ to be the number of 
remaining edges.  We have that $c^+ = d_{\max} - d_{\min} + c^-$.\\
Now note that whenever we have vertices $w_1,w_2,w_3$ such that the path $P$ from $w_{min}$ to $w_{max}$ 
contains an edge from $w_1$ to $w_2$ followed by an edge from $w_3$ to $w_2$, we can do the following. 
We can merge $w_2$ with $s$, remove both edges (as they are now useless), then add an edge from $w_1$ to $w_3$ and 
increase the depth of $w_3$ and all later vertices on the path by 1. This keeps $c^{+}$ the same and reduces $c^{-}$ by 1. 
In this way, we can eliminate all edges going the opposite direction as $P$, and when we are done, we will have a path 
of length $c^{+} \geq d_{\max} - d_{\min}$, which will have at least $d_{\max} - d_{\min} + 1$ vertices.\\
For the general case, note that $H$ will contain at least one such path $P$ as a subgraph and then note that 
we can ignore all of the other other vertices of $H$ by merging them with $s$ or $t$ and then removing useless edges.
\end{proof}

Before handling the case of a general directed tree $G$, we find bounds in the case where $G$ is a flow-out tree, whose bounds are the foundation of our proof of
the general case.

\begin{df}
A \textbf{\emph{flow-out tree}} $G$ is a tree with a special vertex $r$ (the \textbf{\emph{root}})
such that there is a path from $r$ to every other vertex of $G$.
\end{df}

\begin{thm} \label{thm:out_tree}
Let $G$ be a flow-out tree with root $s$ and a path of length $\ell$ from $s$ to $t$.  For $i\ge 1$, define $d_i$ to be the number of
vertices whose descendants have a maximum distance of $i$ from $s$.  Define an additional
sequence $c_1$, $c_2$, $\cdots$, $c_{\lceil \lg \lg \ell \rceil}$ with
the property that $c_1=n$ and for all $i\ge 2$,
\begin{align}c_i&=\sum_{j=2^{2^{i}}}^{n}d_j.\label{eq:2}\end{align}
We then have that
\begin{align}\ell^{\Omega(\lg \ell)}\max_{1\le i \le \lceil \lg \lg \ell
  \rceil}c_i^{\Omega(2^i)}&\le m(\sigma(G))\label{eq:1}\end{align}
\end{thm}

\begin{proof}
To prove the bound, we show that for all $i\le \lceil \lg \lg \ell \rceil$,
$$m(\sigma(G))\ge \ell^{\Omega(\lg \ell)}c_i^{\Omega(2^i)}.$$
Consider the set of vertices which are not on the path from $s$ to $t$.
If we merge these vertices with $s$ and remove useless edges directed
toward $s$, we are left with a single path of length
$\ell$.  From Theorem \ref{thm:merge} and Theorem \ref{thm:many_paths},
we obtain
$$m(\sigma(G))\ge \ell^{\Omega(\lg \ell)}.$$
Thus, it is sufficient to prove that
$$m(\sigma(G))\ge c_i^{\Omega(2^i)}$$
and take the geometric mean of these two bounds.
If $i=1$, take every edge $a\to b$ not on the path from $s$ to $t$ and
change it to $s\to b$.  By Theorems \ref{thm:bbox} and
\ref{thm:repl}, we get
$$m(\sigma(G))\ge n^{\Omega(1)}\ell^{\Omega(\lg \ell)}\ge
c_1^{\Omega(2^1)}$$
Now consider $i\ge 2$, let $k=2^{2^i}$.
Let $S_1$ be the set of vertices $v$ such for any vertex $w$ which is a
descendant of $v$, the distance
from $s$ to $w$ is less than $k$.  Merge the elements of $S_1$ with
$s$.  The remaining vertices have descendants which are a distance of at
least $k$ from $s$.  Thus, there are exactly $c_i$ vertices remaining by
definition of $c_i$.

We split the problem into two cases. Let $\bar d_k$ be the number of vertices with a depth at
most $\lceil k/2\rceil$.

\emph{Case 1:} $\bar d_k\ge \sqrt{c_i}$.

For each vertex $v$ at depth $\lceil k / 2\rceil$ we can choose a path $P_v$ from that vertex whose length is $\lfloor k / 2\rfloor$.  For the vertex $v$ with descendant $t$, we can choose $P_v$ so that it is contained in the path from $v$ to $t$.  For each $v$, let $w_v$ be the other endpoint of $P_v$.  Merge all vertices which are not any path $P_v$ and are not descendants of any $w_v$ with $s$. Merge all descendants of each $w_v$ with $t$.  If a $w_v$ has no descendants, add an edge from $w_v$ to $t$.  The number of vertices at exactly depth $\lceil k / 2\rceil$ is at least $\bar d_k /\lceil k / 2\rceil$, as the number of vertices at depth $i + 1$ is at least the number of verticies at depth $i$.  Each path has $\lceil k / 2\rceil$ vertices, implying that there are at least $\bar d_k$ vertices total.  Thus, from Theorem \ref{thm:merge} and Theorem \ref{thm:many_paths} we obtain
$$m(\sigma(G))\ge \left(\bar d_k\right)^{\Omega(\lg (k/2))}=
(c_i)^{\Omega(2^i)}.$$

\emph{Case 2:} $\bar d_k< \sqrt{c_i}$.
 
Let $D_{\lceil k/2\rceil}$ be the set of points at a distance of $\lceil
k/2\rceil$ from $s$.  Let $v_{\lceil k/2\rceil}$ be the element of
$D_{\lceil k/2\rceil}$ with the largest subtree $H$.  Merge $v_{\lceil k/2\rceil}$ with the vertex
$t$.  We now have two subcases to consider.

\emph{Subcase 1:} $t$ is in $H$.
Merge the path from $v_{\lceil k/2\rceil}$ to
$t$ with $t$.  There may now be subtrees disconnected from $G$.

\emph{Subcase 2:} $t$ is not in $H$.
Let $w$ be the vertex which is at a distance of
$\lceil k/2\rceil$ away from $t$.  Merge $w$ with $s$.

\begin{figure}
\includegraphics[width=2in, bb=0 0 291 420]{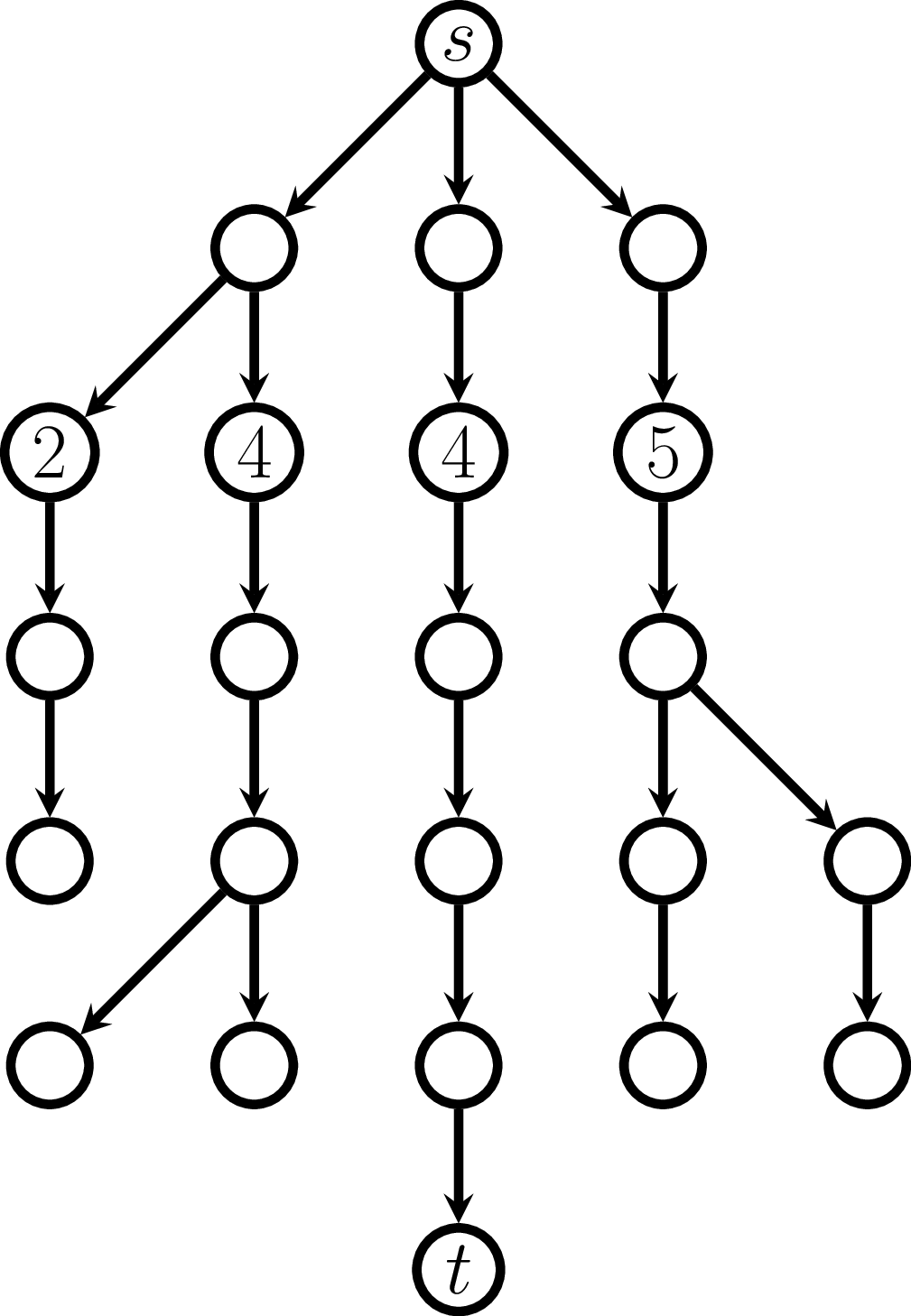}
\hfill
\includegraphics[width=2in, bb=0 0 287 418]{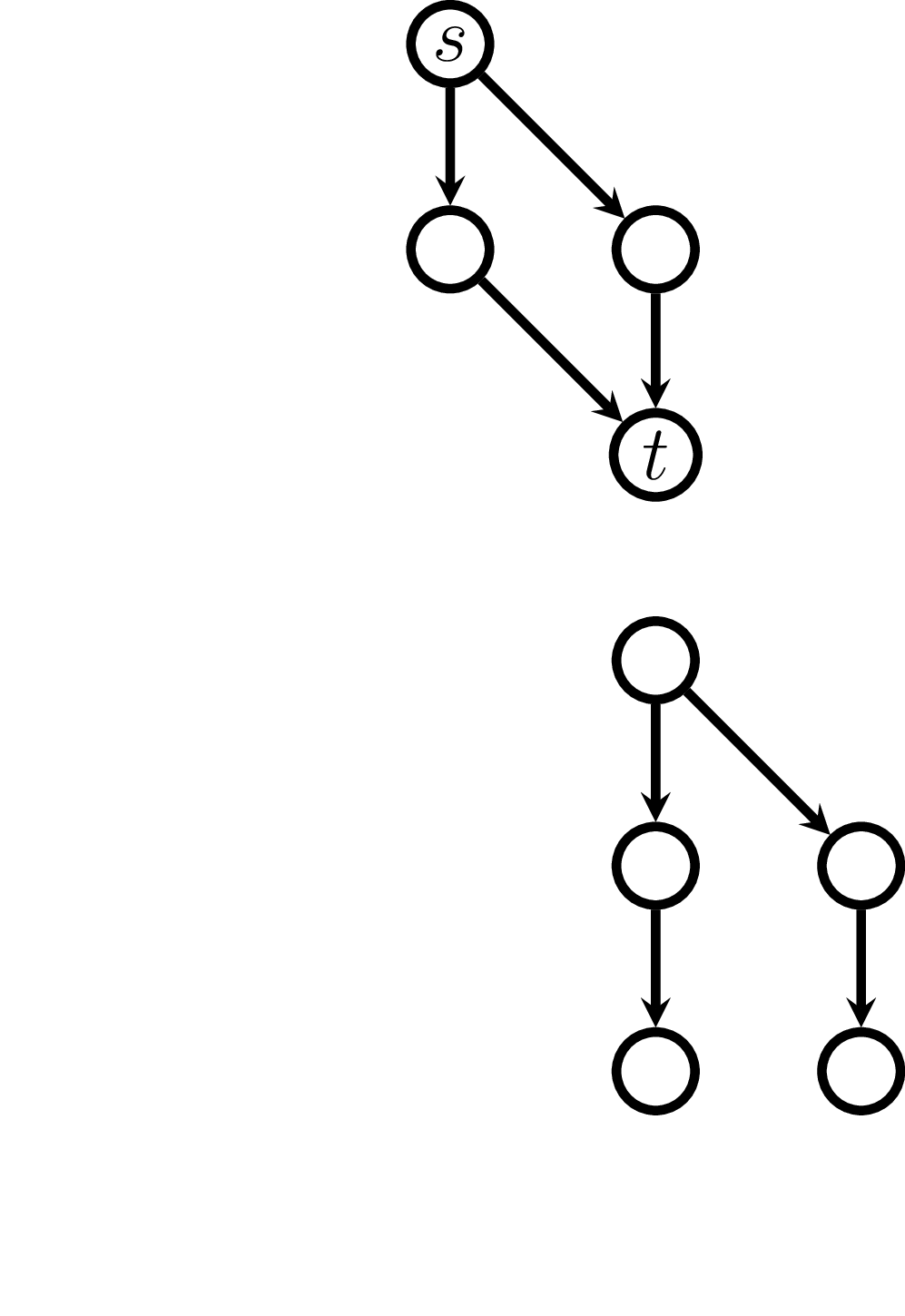}
\hfill
\includegraphics[width=2in, bb=0 0 286 418]{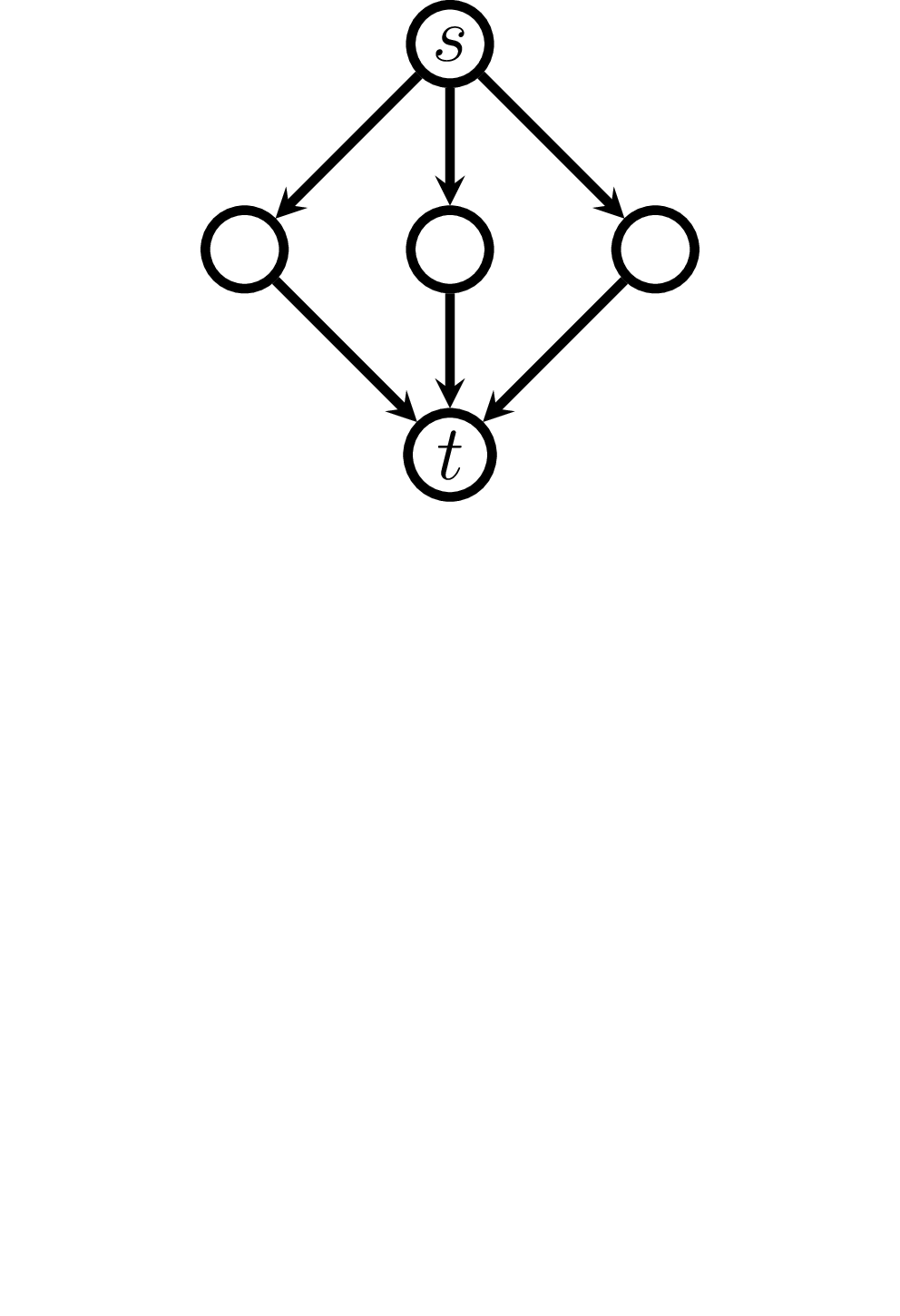}
\caption{Example of the manipulation of $G$ in Case 2, Subcase 2 of the
  proof of Theorem \ref{thm:out_tree}.  In the first graph,
  the numbered vertices are the sizes of their subtrees.  Notice that $k=4$.}
\label{fig:example}
\end{figure}

We now treat both subcases identically.  Let $S_1$ be the set of
vertices which are not
in $H$ and are not on a path from $s$ to
$t$.  Merge $S_1$ with $s$.  Because $H$ is the largest subtree, it has
at least
$$\frac{c_i}{\bar d_k}\ge \sqrt{c_i}$$
vertices.  Some of these vertices were removed in the
subcase where $t\in H$, but at most $\ell - k / 2 + 1$ such vertices were removed. Thus, $G$ has at least
$\sqrt{c_i}-\ell$ vertices.

After removing useless edges from the mergings, we have a collection of disconnected
trees. Let $j$ be the number of such trees and let the sizes of these
trees be
$$a_1, a_2,\hdots, a_j.$$
By Lemma \ref{lem:float}, we can reduce these trees to paths with
$$\lceil \sqrt{a_1}\rceil,~\lceil \sqrt{a_2}\rceil,~\hdots,~\lceil\sqrt{a_j}\rceil$$
vertices, respectively.  We can then add edges and perform mergings to obtain a collection of
paths from $s$ to $t$ of length $\lceil k/2\rceil$ containing at least
$$\frac{\sqrt{\sqrt{c_i}-\ell}}{3}$$
vertices.
From Theorem \ref{thm:many_paths}, we get a lower bound of
$$m(\sigma(G))=
\left(\frac{\sqrt{\sqrt{c_i}-\ell}}{3}\right)^{\Omega(\lg k/2)}=
(\sqrt{c_i}-\ell)^{\Omega(\lg k)}.$$
If $\sqrt{c_i}\ge 2\ell$, then $m(\sigma(G))= c_i^{\Omega(\lg k)}$.
Otherwise, $m(\sigma(G))= \ell^{\Omega(\lg \ell)}$ is a better lower
bound as $\ell\ge k$.

\end{proof}

A dual structure of the flow-out tree is the flow-in tree.

\begin{df}
A \textbf{\emph{flow-in tree}} $G$ is a tree with a special vertex $r$ (the \textbf{\emph{sink}})
such that there is a path to $r$ from every other vertex of $G$.
\end{df}

Theorem \ref{thm:out_tree} has the following corollary which states an
analogous bound for flow-in trees.

\begin{cor}\label{thm:in_tree}
Let $G$ be a flow-in tree with sink $t$ and a path of length $\ell$ from $s$ to $t$.
For $i\ge 1$, define $d_i$ as the number of
vertices whose ancestors, including itself, have a maximum distance of $i$ from $s$.  Define an additional
sequence $c_1$, $c_2$, $\cdots$, $c_{\lceil \lg \lg \ell \rceil}$ with
the property that $c_1=n$ and for all $i\ge 2$, the element $c_i$
satisfies
(\ref{eq:2}). Hence, $m(\sigma(G))$ satisfies the bounds (\ref{eq:1}). 
\end{cor}

\begin{proof}
We reverse the direction of every edge, and swap the labels of $s$ and
$t$.  The obtained tree satisfies the hypothesis of Theorem
\ref{thm:out_tree}.  Thus, the same bounds hold.
\end{proof}

\begin{proof}[Proof of lower bound]

For this graph $G$, let $H_s$ be the graph induced by the set of
vertices $v$ for which a path from $s$ to $v$ exists.  Let $H_t$ be the
graph induced by the set of vertices $v$ from which a path from $v$ to
$t$ exists.  In both cases, we do not include the vertices on the path
from $s$ to $t$.

The proof of this bound is divided into two parts.  The first is to
show that $m(\sigma(G))\ge (\ell+\bar d)^{\Omega(\lg \ell)}.$
The second is to show that $m(\sigma(G))\ge (c_i^s+c_i^t)^{\Omega(2^i)}$ for
all $i$.  We then take the geometric mean of these two bounds.

For the first part, we can merge all of the vertices of $H_s$  with $s$ and all
of the vertices of $H_t$ with $t$.  In the graph $G_{(H_s,H_t)}$, there
may be edges not on the path from $s$ to $t$ which are connected to $s$
or $t$.  These edges are directed towards $s$ or away from $t$,
so they are useless.  We can remove these useless edges to obtain a
graph consisting of a single path from $s$ to $t$ and a collections of
trees which are disconnected from the main path.  Let the size of these $k$
trees be $a_1,~a_2,~\hdots ~,~a_k$.  Notice that $a_1+a_2+\cdots +a_k=\bar
d$.  From Lemma \ref{lem:float}, we can
reduce these trees to paths with $\left\lceil \sqrt{a_1}\right\rceil,~\left\lceil \sqrt{a_2}\right\rceil,~\hdots~,~\left\lceil \sqrt{a_k} \right\rceil$ vertices.
We can link these paths into a long path of length at least
$$\sum_{i=1}^{k}\sqrt{a_i}\ge \sqrt{\bar d}.$$  We can then merge this
long path with the path from $s$ to $t$ to create a collection of
disjoint paths of length $\ell$ from $s$ to $t$.  The total number of
vertices is at least $\sqrt{\bar d}/3+\ell$.
From Theorem \ref{thm:many_paths} we get a lower bound of
$$m(\sigma(G))\ge \left(\frac{\sqrt{\bar d}}{3}+\ell \right)^{\Omega(\lg
  \ell)}= (\bar d+\ell)^{\Omega(\lg \ell)},$$
as desired.

For the second part, merge the all vertices not in $H_s$ nor on the path from $s$ to $t$ with $t$.  We are then
left with a flow-out tree.  By Theorem \ref{thm:out_tree}, $m(\sigma(G))\ge (c_i^s)^{\Omega(2^i)}.$
Using a symmetric argument with a flow-in tree, by Corollary
\ref{thm:in_tree}, we have $m(\sigma(G))\ge (c_i^t)^{\Omega(2^i)}.$
Thus, $m(\sigma(G))\ge (c_i^s+c_i^t)^{\Omega(2^i)},$ as desired.
\end{proof}

\subsection{Upper Bound}

\begin{proof}[Proof of upper bound]
Like in the proof of the lower bound, let $H_s$ be the graph induced by the set of
vertices $v$ for which a path from $s$ to $v$ exists.  Let $H_t$ be the
graph induced by the set of vertices $v$ from which a path from $v$ to
$t$ exists.  In both cases, we do not include the vertices on the path
from $s$ to $t$.

First, take the vertices not in $H_s$, $H_t$, or the path from $s$ to $t,$ and remove any
edges connected to them.  Thus, $G$ is  a flow-in tree, a flow-out
tree, and a collection of disconnected vertices.  We use a construction very
similar to that used in the proof of Theorem \ref{thm:out_tree}.  

\begin{figure}
\begin{tabular}{MM}
\Huge {$G_1$} & \includegraphics[width = 4in, bb=0 0 652 141]{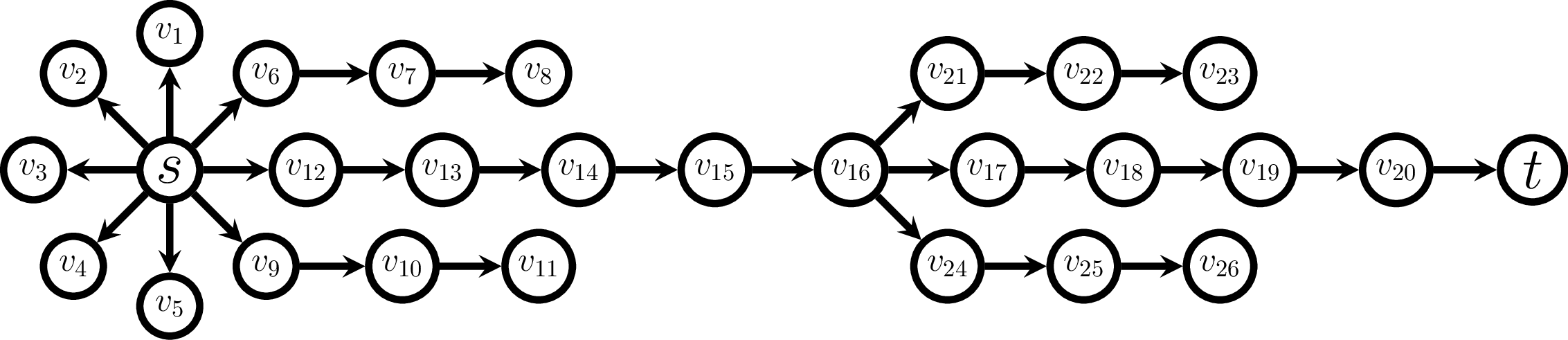}\\
\Huge {$G_2$} & \includegraphics[width = 4in, bb=0 0 652 141]{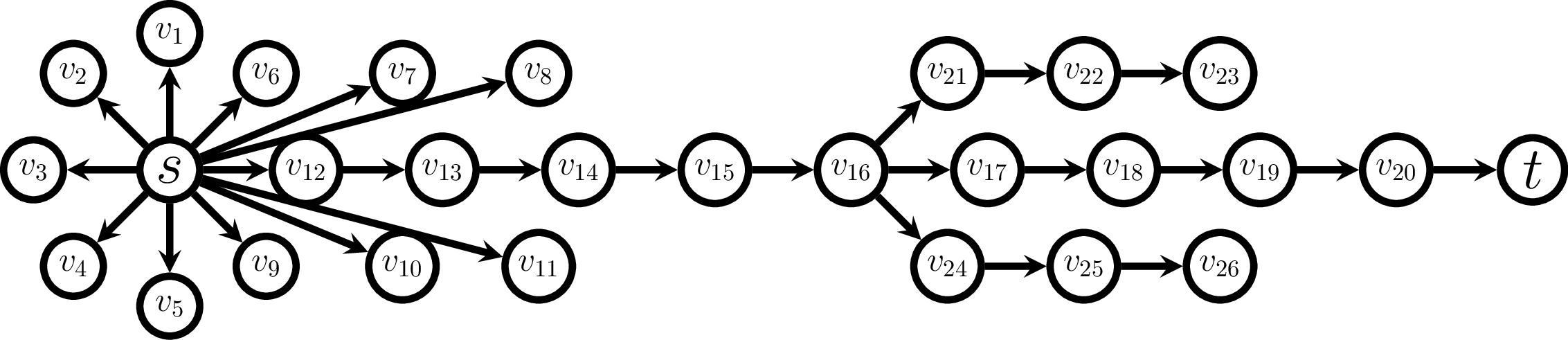}\\
\Huge {$G_3$} & \includegraphics[width = 4in, bb=0 0 652 225]{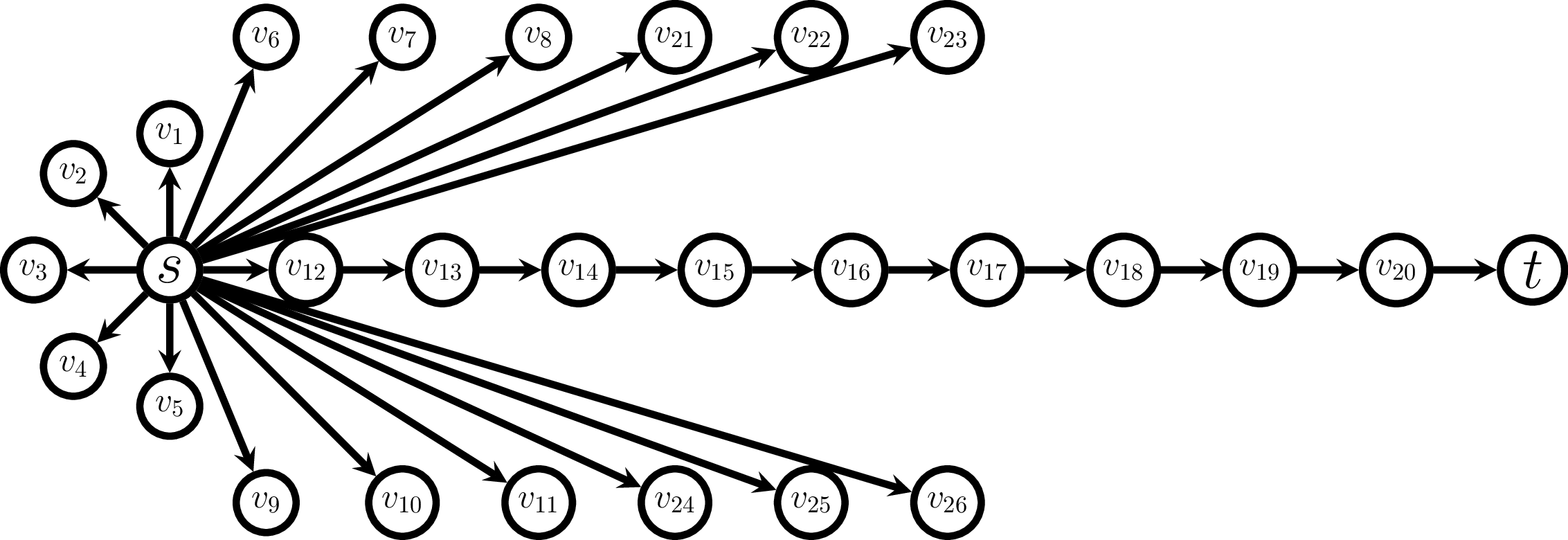}
\end{tabular}

\caption{Example of a sequence of graphs $G_1,G_2,G_3$ used in proving the upper bound
  in Theorem \ref{thm:out_tree}.  Note that in this example, $G = G_1$.}
\label{fig:up}
\end{figure}

Let $\bar c$ be the number of points which are at a distance of more than
$\ell$ from $s$ or to $t$.  Construct a new graph $G_1$ where all the edges
connected to these $\bar c$ points are removed.  We now construct
graphs $G_2,~\hdots,~G_{\lceil \lg \lg \ell\rceil+1}$ inductively as
follows.  Given $G_{i-1}$, let $P^s_i$ be the set of vertices which are at a
distance of less than $2^{2^{i-1}}$ from $s$ and do not have any
children.  Let $P^t_i$ be the set of vertices which are at a distance of
less than $2^{2^{i-1}}$ to $t$ and do not
have any ancestors.  Also, add to $P_i^s$ any vertices which are on the
paths from $s$ to $P_i$; and add to $P_i^t$ any vertices
which are on paths from $P_i^t$ to $t$.  To create $G_i$, remove the edges connecting $P_i^s$ to $G_{i-1}$
and add edges directly from $s$ to $P_i^s$.  Similarly, remove edges
connected $P_i^t$ to $G_{i-1}$ and add edges directly to $t$. From the definition of $c_i^s+c_i^t$,
there are at least $n-c_i^s-c_i^t$ vertices which are directly connected to $s$
or $t$. Refer to Figure \ref{fig:up} for an
example of this construction.  

In the graph $G_{\lceil \lg \lg \ell \rceil+1}$, every vertex is
directly connected via a single edge to $s$ or $t$, on the path from $s$ to
$t$, or disconnected from the graph.  From Theorem \ref{thm:bbox}, there
is a monotone switching network $G'_{\lceil \lg \lg \ell \rceil+1}$ of size
$n^{O(1)}(\ell+\bar c+\bar d)^{O(\lg \ell)}$
accepting the elements of $\sigma(G_{\lceil \lg \lg \ell \rceil+1})$.

Next, we construct inductively the sequence of sound monotone switching
networks $G'_{\lceil \lg \lg \ell \rceil},~\hdots,~G'_2,~G'_1$
such that $G'_{i}$ accepts the elements of $m(\sigma(G_i))$.  In the
monotone switching network $G'_{i+1}$, consider an edge with a label of
the form $s\to a$.  Consider a graph $\bar{G}_{i+1}\in \sigma(G_{i+1})$
which crosses that edge in its accepting path.  The corresponding graph
$\bar{G}_i$ may not be able to cross that edge because the edge was
deleted and became a path of length at most $2^{2^{i-1}}$ in $\bar G_i$ from $s$ to $a$.  Thus, to construct $G'_i$, we replace that
edge with a monotone switching network which checks if there is a path
of length at most $2^{2^{i-1}}$ from $s$ to $a$, assuming at most $c_i^s+c_i^t+\bar d +\ell$ vertices are not lollipops.

To construct the monotone switching network, we take the monotone switching network of size $n^{O(1)}(c_i^s+c_i^t+\bar d +\ell)^{O(2^i)}$ guaranteed by Theorem \ref{thm:bbox}, which checks if there if there is a path of length at most $2^{2^{i-1}}$ from $s$ to $t$ on $V(G\setminus \{a\})\cup \{s,t\}$, assuming at most $c_i^s+c_i^t+\bar d +\ell$ vertices are not lollipops.  Now for each edge of the switching network, whose label is of the form $v\to t$, replace it with two parallel edges, one of the form $v\to a$ and the other of the form $v\to t$.  It is easily verified this modified monotone switching network does indeed check whether there is a path from $s$ to $a$ or a path from $s$ to $t$, given that at most $c_i^s+c_i^t+\bar d +\ell$ vertices are not lollipops.  We can deal with edges in $G_{i+1}'$ with labels of the form $a\to t$ in a similar way.  The number of edges in each such checker is at most
$$n^2\left(n^{O(1)}(c_i^s+c_i^t+\bar
  d)^{O(2^i)}\right)^2=n^{O(1)}(c_i^s+c_i^t+\bar d+\ell)^{O(2^i)}.$$
Thus, the number of edges of $G_i'$ is at most
$$|E(G_i')|\le |E(G_{i+1}')|n^{O(1)}(c_i^s+c_i^t+\bar d+\ell)^{O(2^i)}.$$
Thus,
\begin{align}
V(G')\le E(G')\le E(G_1') &=E(G_{\lceil \lg \lg \ell\rceil +1}')\prod_{i=1}^{\lceil
                \lg \lg \ell\rceil}n^{O(1)}(c_i^s+c_i^t+\bar
              d+\ell)^{O(2^i)}\notag \\
              &=n^2(n^{O(1)}(\ell+\bar c+\bar d)^{O(\lg \ell)})^2n^{O(\lg \lg
                \ell)}\prod_{i=1}^{\lceil \lg \lg
                \ell\rceil}(c_i^s+c_i^t+\bar d+\ell)^{O(2^i)}\notag \\
              &=n^{O(\lg \lg \ell)}(\ell+\bar c+\bar d+\ell)^{O(\lg \ell)}\prod_{i=1}^{\lceil \lg \lg
                \ell\rceil}(c_i^s+c_i^t+\bar d)^{O(2^i)}.\label{eq:end}
\end{align}
From its definition, $\bar c \le c_i^s+c_i^t$, for all $i$. Therefore when $i = \lg \lg \ell$, we have $$\bar c^{O(\lg \ell)}\le (c_i^s+c_i^t+\bar d+\ell)^{O(2^i)},$$
and $$\prod_{j=1}^{\lceil \lg \lg \ell\rceil}\bar d^{O(2^j)}\le \bar d^{O(\lg \ell)}.$$
Hence, inequality (\ref{eq:end}) is equivalent to
$$m(\sigma(G))\le n^{O(\lg \lg \ell)}(\ell+\bar d)^{O(\lg \ell)}\prod_{i=1}^{\lceil \lg \lg
                \ell\rceil}(c_i^s+c_i^t)^{O(2^i)}.$$
\end{proof} 

\section{Conclusion} 

Sound monotone switching networks provide an insightful way of analyzing
monotone computation.  Previously, Potechin~\cite{Potechin1, Potechin2} found tight bounds in the case where the inputs were the permutation sets of very specific kinds of trees
and acyclic graphs.  From these
earlier results, we proved in Theorem \ref{thm:all_tree} nearly tight bounds for all directed
trees.  These bounds give us insight into the structure of space-efficient
monotone computation.  From Theorem \ref{thm:all_tree}, the exponent
for $c_1^s+c_1^t$ is orders of magnitude smaller than the exponent for $c_{\lceil \lg
  \lg \ell\rceil}^s+c_{\lceil \lg \lg \ell \rceil}^t$.  We can infer
from this that monotone computation is more effective at analyzing vertices
closer to $s$ and $t$ than vertices which are farther.  This
suggests that the optimal algorithm for ST-connectivity in a monotone
computation model is akin to a breadth-first search.

Possibilities of future investigation include:

\begin{itemize}
\item Generalize the bounds to permutation sets of
all acyclic graphs and eventually all graphs.
\item Improve known bounds.  Currently, these bounds are
  within a factor of $O(\lg \lg \ell)$ in the exponent.  Can this be
  improved to a factor of $O(1)$ in the exponent?
\item Find algorithms corresponding to these monotone switching
  networks.  The existence of a monotone switching network of size $m$ heuristically implies that an algorithm with $O(\log m)$ memory use
  exists, but such an algorithm may not necessarily
  exist. Much work can be devoted to determining
  whether or not these algorithms exist and finding elegant
  implementations if they indeed exist. 
\item Extend these results to non-monotone switching networks.  These
  more general structures account for all possible classical
  computations. Obtaining tight bounds in this case would solve the open log-space versus
  nondeterministic log-space problem.   
\end{itemize}

\section{Acknowledgments} 
We would like to thank Dr. Tanya
Khovanova for her helpful suggestions pertaining to this research;
Dr. Pavel Etingof for organizing the RSI mathematics research; Dr. John
Rickert, Sitan Chen, Dr. Sendova, and other RSI staff members and students for giving
suggestions to improve our paper.  We would also like to acknowledge the Massachusetts Institute
of Technology for hosting the Research Science Institute; Joshua Brakensiek's Research Science Institute sponsors
Dr. and Mrs. Daniel Dahua Zheng and the American Mathematical Society
for funding the Research Science Institute; D.E. Shaw \& Co., L.P. for
naming him a D.E. Shaw scholar; and the MIT Mathematics
Department and the Center for Excellence in Education, for making this
paper possible.  We would also like to thank Mr. Joseph DiNoto, for recommending TikZ for the graphics of this paper.  Joshua Brakensiek would also like to thank his parents, Warren and
Kathleen, for their support.

\nocite{*}
\bibliographystyle{plain}

\bibliography{biblio}


\appendix

\section{Arbitrary Merging May Not Decrease $m(\sigma(G))$}\label{app:merge}

Recall from Theorem \ref{thm:merge}, that merging vertices of $G$ with
$s$ or $t$ does not increase the value of $m(\sigma(G))$.  We now
exhibit a construction which shows how
merging arbitrary vertices fails in general.

Let $\ell$ be a positive integer, construct a graph $G$ with $\ell-1$ special
sets $C_{1}$, $C_2$, $\hdots$, $C_{\ell-1}$.  Let each set contain
$\left\lfloor \frac{n-2}{2(\ell-1)}\right\rfloor$ vertices.  Let there be no edges between
any of the vertices in each $C_i$, but let all of the vertices in each
$C_i$ where $i< \ell-1$, be directed to all the vertices in $C_{i+1}$.
Additionally, let there be edges from $s$ to all the vertices of $C_1$
and edges from all vertices of $C_{\ell-1}$ to $t$. Let the other
approximately $n/2$ vertices of $G$ not already mentioned be disconnected
from the graph and from each other.

Let $H$ be the graph where for each $i\ge 1$, the vertices of $C_i$ are
merged to a single vertex $v_i$.  Notice that $H$ is merely a path of length with $\ell$ edges from $s$ to $t$ and has $n/2 + \ell - 1$ vertices.

\begin{thm}\label{thm:merge_fail}
If $n > 20$ and $n/4 > \ell$ then,
\begin{align*}
m(\sigma(G))&= n^{O(1)}\ell^{O(\ell)}\\
m(\sigma(H))&= n^{\Omega(\lg \ell)}
\end{align*}
\end{thm}

\begin{proof}
Since $\sigma(H)$ is the set of minimal elements of $\mathcal P_{\ell}$, Theorem \ref{thm:all_len} implies that $$m(\sigma(H)) = \left(\frac{n}{2} + \ell - 1\right)^{\Omega(\lg \ell)} = n^{\Omega(\lg \ell)}.$$
To obtain the bound for $m(\sigma(G))$, we will construct the monotone switching network $G'$ as follows.  Let $$p = \left(\frac{1}{4(\ell - 1)}\right)^{\ell - 1}.$$  Construct $C = n^2/p$ internally-disjoint undirected paths of length $\ell$ between $s'$ and $t'$.  Each such path has consecutive edge labels of the form
$$s\to w_1, \,\, w_1\to w_2, \,\, \hdots, \,\, w_{\ell - 1}\to t.$$
Where each vertex $w_i$ is selected uniformly and at random.  For any particular graph $\bar{G} \in \sigma(G)$, the probability that a particular path of $G'$ accepts $\bar G$ is, if $\ell$ is sufficiently large,
$$\left(\frac{\left\lfloor \frac{n-2}{2(\ell-1)}\right\rfloor}{n}\right)^{\ell - 1} \ge \left(\frac{1}{4(\ell - 1)}\right)^{\ell - 1} = p.$$
Thus, the probability that $\bar G$ is rejected is at most $1 - p$.  This implies that the expected number of elements of $\sigma(G)$ rejected by $G'$ is at most
$$|\sigma(G)|(1-p)^C \le n!e^{-Cp} = n!e^{-n^2} < 1$$
Thus, by the probabilistic method, there exists a choice of $C = n^2/p$ paths which accept all the elements of $\sigma(G)$.  The size of $G'$ in this case is
$$|V(G')| \le C \ell = n^2 \ell / p = n^2 \ell (4(\ell - 1))^{(\ell - 1)} = n^{O(1)}\ell^{O(\ell)}$$
as desired. 
\end{proof}

Theorem \ref{thm:merge_fail} shows that in the case $n$ is arbitrarily large,
and $\ell$ is a sufficiently large constant,
$m(\sigma(G))<m(\sigma(H))$.  This demonstrates that arbitrary merging
does not necessarily decrease the value of $m(\sigma(G))$.

\section{Improved Bounds on Lemma 5.2} \label{app:float}

Although the bound of $\lceil \sqrt{V(H)}\rceil $ on the path length given in Lemma
\ref{lem:float} is strong enough for proof of Theorems
\ref{thm:out_tree} and \ref{thm:all_tree}, a stronger bound has been found.  This
results improves $\lceil \sqrt{V(H)}\rceil $ to approximately $V(H)/(2\lg V(H))$.  First, we give a definition of the quantity we are bounding.

\begin{df}
Given a graph $H$, define a \textit{\textbf{disconnected-path family}} $\mathcal P = \{P_1, \hdots, P_k\}$ to be a family of disjoint paths of $H$ such that for any two paths $P_i$ and $P_j$, there is no unidirectional path connecting these two paths.   
\end{df}

\begin{df}
Define the \textit{\textbf{size}} of a disconnected-path family $\mathcal P$ to be the total number of vertices in the family.  This quantity is denoted by $V(\mathcal P)$.
\end{df}

\begin{df} \label{df:dp_len}
Define the \textit{\textbf{disconnected-path length}} of a directed tree $H$ to be the maximal-size disconnected-path family of $H$.  We denote this quantity by $p(H)$.
\end{df}

We now prove a theorem which connects Definition \ref{df:dp_len} to Lemma \ref{lem:float}.

\begin{lem} \label{lem:msn_to_dp_len}
Let $G$ be a directed graph and $H$ be a directed tree disconnected from $G$.  Let $P$ be a path with $p(H)$ vertices.  Then, $$m(\sigma(G\cup H))\ge m(\sigma(G\cup P))$$
\end{lem}

\begin{proof}
In this proof, we construct sets $S$ and $T$ such that the vertices of $S$ are to be merged with $s$ and the vertices of $T$ are to be merged with $t$.  Let $\mathcal P = \{P_1, \hdots, P_k\}$ be a disconnected-path family of size $p(H)$.  We let $v$ be any vertex of $H$ not part of any path of $\mathcal P$.  We define $v$ to be \textit{inbound} if there exists a path $P_i$ such that there is a path from $v$ to some vertex of $P_i$.  Analogously, we define $v$ to be \textit{outbound} if there exists a path $P_i$ such that there is a path from some vertex of $P_i$ to $v$.  We now have four cases to consider

\emph{Case 1: $v$ is neither inbound nor outbound.} 

This is not possible, as $\mathcal P \cup \{\{v\}\}$ would be a disconnected-path family of size greater than $p(H)$.

\emph{Case 2: $v$ is both inbound and outbound.}

This is also not possible, as then there would exist paths $P_i$ and $P_j$ such that there is a path from one vertex of $P_i$ to some other vertex of $P_j$ via $v$.  This path contradicts the the definition of a disconnected-path family.

\emph{Case 3: $v$ is only inbound.}

Let $v$ be an element of $T$.

\emph{Case 4: $v$ is only outbound.}

Let $v$ be an element of $S$.

From these cases, we have constructed the sets $S$ and $T$.  If we performed the mergings, by definition of inbound and outbound, the only new edges connected to $s$ are directed toward $s$ and the only new edges connected to $t$ are directed from $t$.  Hence, all created edges are useless.  Removing these edges, we are left with the paths of $\mathcal P$ disconnected from $G$.  Linking these paths with additional edges, we obtain a path $P$ with $p(H)$ vertices.  Thus it follows from Theorems \ref{thm:useless} and \ref{thm:merge} and Corollary \ref{thm:add2}, $$m(\sigma(G\cup H))\ge m(\sigma(G\cup P)).$$
\end{proof} 

\begin{lem} \label{lem:awesome}
$$p(H)\ge \left\lceil \frac{V(H)}{\lg V(H)+1}\right\rceil.$$
\end{lem}
\begin{proof}
Let $r$ be the root of $H$.  Give the vertices of $H$ a labeling $j : V(H) \to \mathbb Z^+$ such that for each vertex $v$:
\begin{itemize}
\item If $v$ has no children, then $j(v) = 1$.
\item If $v$ has children $w_1,\hdots, w_c$, and $j(w_i)$ has a unique maximum, then let $$j(v) = \max_{1\le i \le c}j(w_i).$$
\item If $v$ has children $w_1,\hdots, w_c$, but multiple children have the same maximal $j$, then let $$j(v) = 1+\max_{1\le i\le c}j(w_i).$$
\end{itemize}

We will now show that $j(v) \le \lg V(H) +1$ for all $v$.  From the definition of $j$, we have that if $j(v) > 1$, then there exist at least two descendants of $v$, $w_1$ and $w_2$ such that $$j(w_1) =j(w_2)= j(v)-1.$$  By induction, if $j(v) > k$, then there exist at least $2^k$ descendants of $v$, $w_1,\hdots, w_2^k$ such that $j(w_i) = j(v)-k$ for all $i$.  We can thus see that
$$V(H) \ge 2^{j(v)-1}.$$
Therefore, $j(v) \le \lg V(H)+1$.

Define $k_i$ to be the number of vertices $w$ such that $j(w) = i$.  Note that for each $i$, the vertices such that $j(w) = i$ form a disconnected-path family.  Thus $p(H)\ge k_i$ for all $i$. By the pigeonhole principle, there exists an $i$ such that
$$p(H) \ge k_i \ge \left\lceil\frac{V(H)}{\lg V(H) + 1}\right\rceil,$$
as desired. 
\end{proof}

\begin{cor}
Let $G$ be a graph and $H$ be a flow-out tree disconnected from $G$.
Let $P$ be a path of length $\lceil V(H)/(\lg V(H)+1)\rceil$.  Then
$$m(\sigma(G\cup H))\ge m(\sigma(G\cup P)).$$
\end{cor}

\begin{cor}
Let $H$ be a flow-in tree.  Then $$p(H) \ge \left\lceil \frac{V(H)}{\lg V(H)+1}\right\rceil.$$
\end{cor}

\begin{cor} \label{lem:pH_disjoint}
Let $H$ be a disjoint collection of flow-in and flow-out trees. Then $$p(H) \ge \left\lceil \frac{V(H)}{\lg V(H)+1}\right\rceil.$$
\end{cor}
\begin{proof}
This inequality follows from the facts that the function $$\left\lceil\frac{x}{\lg x + 1}\right\rceil$$ is subadditive.
\end{proof}

It turns out it is easy to find $p(H)$ exactly, when $H$ is a flow-out tree.  For each vertex $v$, define $d(v)$ to be the maximum number of nodes on
a path from $v$ to some other node.  Define a another function $b(v)$
with the following properties:
\begin{itemize}
\item If there are no edges leading out of $v$, then $b(v)=d(v)=1$.
\item Otherwise, let $w_1,\hdots, w_c$ be the nodes leading out of $v$,
  then
$$b(v)=\max\left(\sum_{i=1}^{c}b(w_i),d(v)\right).$$
\end{itemize}
Let $r$ be the vertex from which there is a path to every other vertex.
We now prove the following three lemmas concerning the value of
$b(r)$.
\begin{lem} \label{lem:big}
$b(r) \ge p(H)$.
\end{lem}
\begin{proof} 
We will prove this by induction on $|V(H)|$.  If $|V(H)| = 1$, then $b(r) = p(H) = 1$.  Assume that $|V(H)| > 1$. We shall show for any disconnected-path family $\mathcal P$, that $b(r)\ge V(\mathcal P)$. Let the children of $r$ be $w_1, \hdots, w_c$.  We have two cases to consider.

\emph{Case 1: $r$ is an element of $\mathcal P$}

In this case, $\mathcal P$ must be a path.  Thus, $b(v) \ge d(v) = V(\mathcal P).$

\emph{Case 2: $r$ is not an element of $\mathcal P$}

Since all the subtrees of $r$ are disconnected from the each other, by the inductive hypothesis, the maximal size of $\mathcal P$ is at most
$$\sum_{i=1}^{c}b(w_i) \le b(v),$$
as desired.

In either case $b(r) \ge V(\mathcal P)$; therefore $b(r) \ge p(H)$.
\end{proof}

\begin{lem}\label{lem:small}
$p(H) \ge b(r)$.
\end{lem}
\begin{proof}
We prove this lemma by constructing a disconnected-path family $\mathcal P$ of size $b(r)$.
Consider a topological ordering of the vertices of $H$,
$$r=v_1,~v_2,~\hdots~, ~v_{V(H)}$$
with the property that if $j<i$ then there is no path from $v_i$
to $v_j$.  We now scan through the list and determine which
elements go in $\mathcal P$.  This is decided as follows.

\begin{itemize}
\item If there is a path from any element already in $\mathcal P$ to $b_i$ or $b(v_i)=\displaystyle\sum_{i=1}^{c}b(w_i)$, then do not add $v_i$ to $\mathcal P$.
\item Otherwise, take the longest path $P_i$ from $v_i$ and add this path to $\mathcal P$. 
\end{itemize}
\begin{figure}
\begin{center}
\includegraphics[width=2in, bb=0 0 293 384]{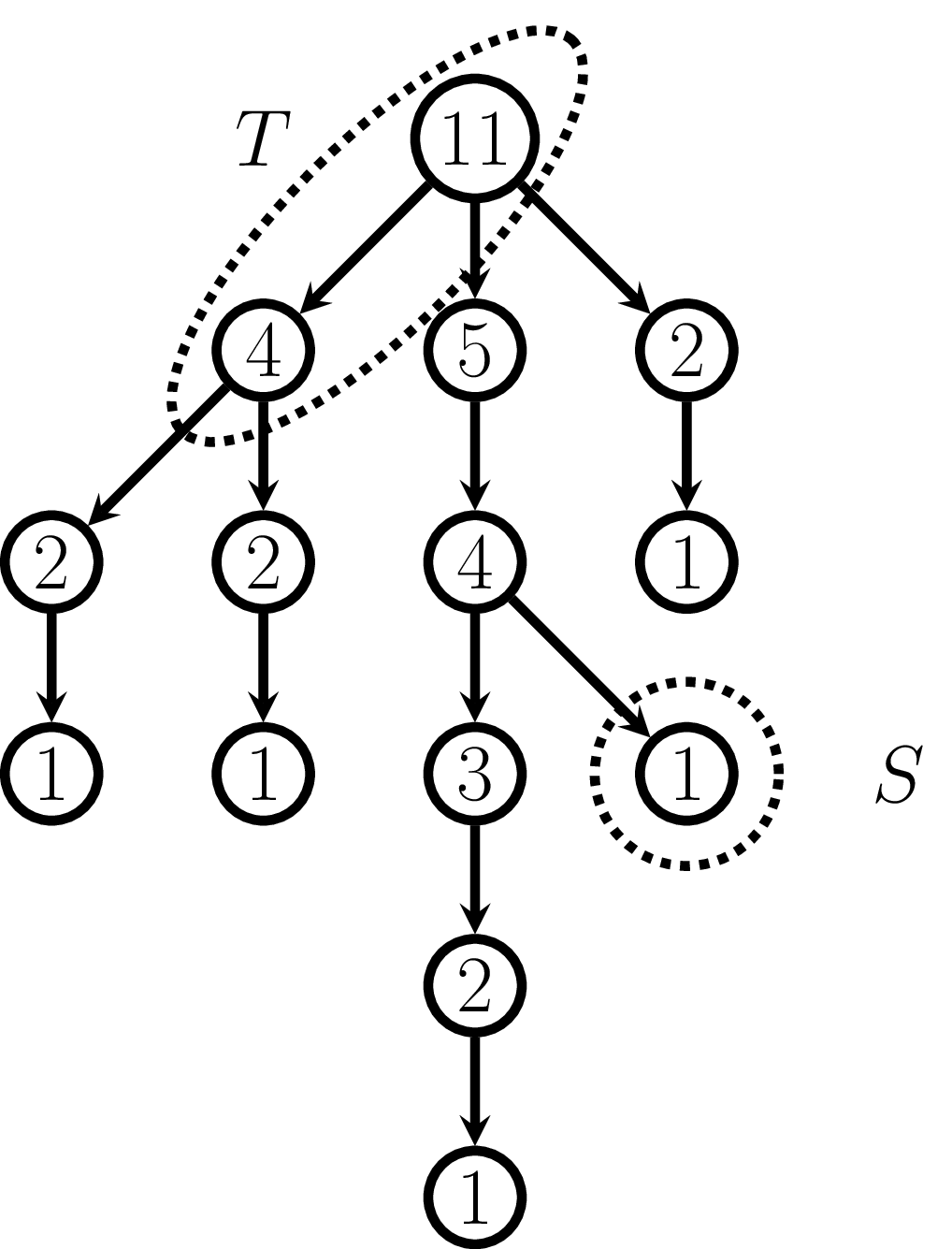}
\end{center}
\caption{Values of $b(v)$ for an example tree.  Also noted are the
  vertices which are elements of $S$ and $T$ from Lemma \ref{lem:msn_to_dp_len}.}
\label{fig:best}
\end{figure}
See Figure \ref{fig:best} for an example.  
From the recursion, we know that $b(r)$ vertices are in $\mathcal P$.  For each path of $H$ which is preserved, all the
vertices leading out of it and leading into it are not in $\mathcal P$.  Hence, $\mathcal P$ is a disconnected-path family and $p(H) \ge b(v)$.
\end{proof}

\begin{note}
By Lemmas \ref{lem:big} and \ref{lem:small}, we have that $b(r) = p(H)$.
\end{note}

From these Lemmas, $p(H)$ can be computed exactly in linear time using the recursive of $b(r)$. 
\begin{rem}
Utilizing this recursion to compute $p(H)$, we can construct flow-out trees which show that Lemma \ref{lem:awesome} is asymptotically optimal.  See figure \ref{fig:optimal} for an example.
\end{rem}
\begin{figure}
\begin{center}
\includegraphics[height = 3in, bb=0 0 282 482]{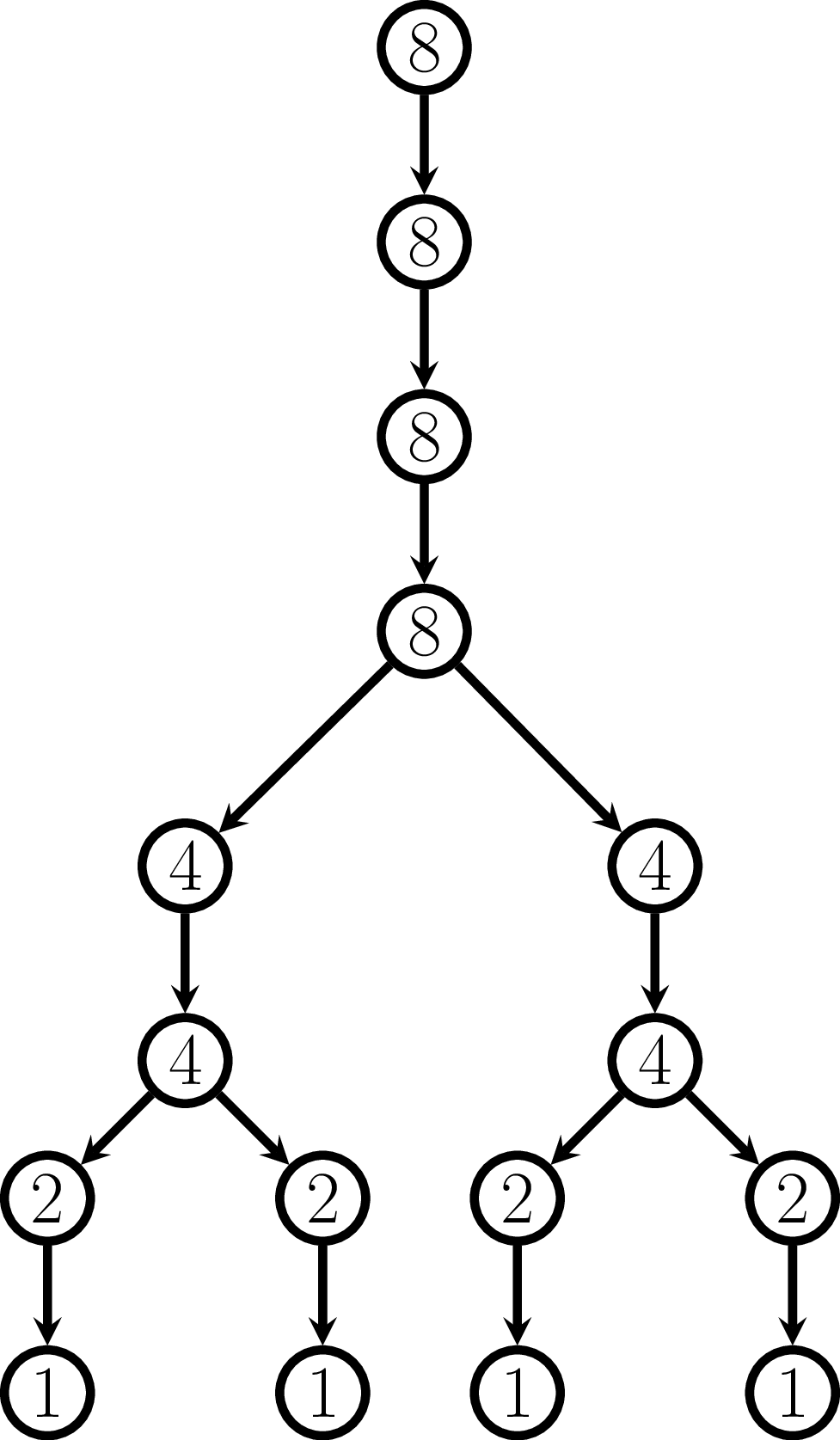}
\end{center}
\caption{Example of a flow-out tree $H$ (with $b(v)$ values added) such that a generalized construction of $H$ demonstrates Lemma \ref{lem:awesome} is asymptotically optimal.}
\label{fig:optimal}
\end{figure}

Now we consider the case that $H$ is a general directed tree.

\begin{lem} \label{lem:awesome2}
Let $H$ be a directed tree.  Then, $$p(H)\ge \left\lceil\frac{V(H)}{2(\lg V(H) + 1)}\right\rceil.$$
\end{lem}

\begin{proof}
Let $r$ be an arbitrary vertex with an indegree of $0$.  Let $B_0 = \{r\}$.  Let $B_1$ be the set of all vertices such that there exists a path from $r$ to any vertex of $B_1$.  Let $B_2$ be the set of all vertices not in $B_0\cup B_1$ such that there exists a unidirectional path from $B_2$ to some vertex of $B_1$.  More generally, let $B_k$ be the set of all vertices not in $B_0\cup B_1\cup \cdots \cup B_{k-1}$ such that there exists a unidirectional path from each vertex of $B_k$ to $B_{k-1}$.  Thus, if $i - j\ge 2$, then the vertices of $B_i$ and $B_j$ are disconnected. Define the two sets \begin{align*}\mathcal B_1 &= \bigcup _{i\ge 0}B_{2i+1}\text{\;\;\;and}\\\mathcal B_2 &= \bigcup _{i\ge 0}B_{2i}.\end{align*}  It is apparent that $\mathcal B_1$ is a disjoint collection of flow-out trees and that $\mathcal B_2$ is a disjoint collection of flow-in trees.  From Corollary \ref{lem:pH_disjoint}, we have that
$$p(\mathcal B_i)\ge \left\lceil\frac{V(\mathcal B_i)}{\lg V(\mathcal B_i) + 1}\right\rceil.$$
Hence,
\begin{align*}p(H) &\ge \max(p(\mathcal B_1), p(\mathcal B_2))\\ &\ge \left\lceil\frac{V(H)/2}{\lg (V(H)/2) + 1}\right\rceil\\ &\ge \left\lceil\frac{V(H)}{2(\lg V(H) + 1)}\right\rceil,\end{align*}
as desired.
\end{proof}

\begin{thm}\label{thm:algo}
There exists a linear-time algorithm, given a directed tree $H$ as input, which computes $p(H)$.
\end{thm}

\begin{proof}
Let $r$ be an arbitrary vertex of $H$.  Also, let $n=V(H)$. Consider an arbitrary ordering
$$v_1 = r,v_2,\hdots, v_n$$
of the vertices of $H$ such that for any undirected path from $r$, $$(r, v_{i_1}, w_{i_2}, \hdots, w_{i_k}),$$ the indices $i_1$ to $i_k$ are in increasing numerical order.  For example, this ordering could be a level-order traversal of $H$ rooted at $r$, where the direction of the edges are ignored.  Notice that this implies at most one vertex connected to a vertex $v_j$ has an index less than $j$.

Define $V_1,V_2,\hdots, V_n$ to be subgraph of $H$ such that $V_i$ has as vertices the $v_j$ such that $j\ge i$ and there exists an undirected path from $v_j$ to $v_i$ using only the vertices $v_k$ such that $k \ge i$.  Notice that $V_1 = H$.

We now define the six functions $a_1, a_2, a_3, a_4, a_5, a_6$ from the set of vertices of $H$ to the positive integers.  There definitions are as follows.
\begin{itemize}
\item $a_1(v_i)$ is the maximal-size disconnected-path family $\mathcal P^1_i$ of $V_i$ such that $v_i$ is not a vertex of a path of $\mathcal P^1_i$ and there is no directed path from $v_i$ to any element of $\mathcal P^1_i$. 
\item $a_2(v_i)$ is the maximal-size disconnected-path family $\mathcal P^2_i$ of $V_i$ such that $v_i$ is not a vertex of a path of $\mathcal P^2_i$ and there is no directed path from any element of $\mathcal P^2_i$ to $v_i$. 
\item $a_3(v_i)$ is the maximal-size disconnected-path family $\mathcal P^3_i$ of $V_i$ such that $v_i$ is a vertex of a path of $\mathcal P^3_i$ and that this path has no edge directed from $v_i$.
\item $a_4(v_i)$ is the maximal-size disconnected-path family $\mathcal P^4_i$ of $V_i$ such that $v_i$ is a vertex of a path of $\mathcal P^4_i$ and that this path has no edge directed toward $v_i$.
\item $a_5(v_i)$ is the maximal-size disconnected-path family $\mathcal P^5_i$ of $V_i$ such that $v_i$ is a vertex of a path of $\mathcal P^5_i$.
\item $a_6(v_i)$ equals $p(V_i)$
\end{itemize}
Thus, $a_6(v_1)$ equals $p(H)$.
\begin{lem}
These six functions satisfy the following recursion.
\begin{align*}
a_1(v_i)={}&\sum_{\substack{(v_i,v_j)\in E(H)\\i\le j}}a_1(v_j)+\sum_{\substack{(v_j,v_i)\in E(H)\\i\le j}}a_6(v_j)\\
a_2(v_i)={}&\sum_{\substack{(v_j,v_i)\in E(H)\\i\le j}}a_2(v_j)+\sum_{\substack{(v_i,v_j)\in E(H)\\i\le j}}a_6(v_j)\\
a_3(v_i)={}&1+\sum_{\substack{(v_i,v_j)\in E(H)\\i\le j}}a_1(v_j)+\sum_{\substack{(v_j,v_i)\in E(H)\\i\le j}}a_2(v_j)\\ &+ \max\left(0,\max_{\substack{(v_i,v_j)\in E(H)\\i\le j}}(a_3(v_j)-a_1(v_j))\right)\\
a_4(v_i)={}&1+\sum_{\substack{(v_j,v_i)\in E(H)\\i\le j}}a_2(v_j)+\sum_{\substack{(v_i,v_j)\in E(H)\\i\le j}}a_1(v_j)\\ &+ \max\left(0,\max_{\substack{(v_j,v_i)\in E(H)\\i\le j}}(a_4(v_j)-a_2(v_j))\right)\\
a_5(v_i)={}&1+\sum_{\substack{(v_i,v_j)\in E(H)\\i\le j}}a_1(v_j)+\sum_{\substack{(v_j,v_i)\in E(H)\\i\le j}}a_2(v_j)\\ &+ \max\left(0,\max_{\substack{(v_i,v_j)\in E(H)\\i\le j}}(a_3(v_j)-a_1(v_j))\right)\\ &+ \max\left(0,\max_{\substack{(v_j,v_i)\in E(H)\\i\le j}}(a_4(v_j)-a_2(v_j))\right)\\
a_6(v_i)=&\max(a_1(v_i),a_2(v_i),a_5(v_i))
\end{align*}
\end{lem}
\begin{proof}
For each vertex $v_i$, we let $D_i^+$ be the set of vertices connected to $v_i$ by an edge directed away from $v_i$ and $D_i^-$ be the set of vertices connected to $v_i$ by an edge directed toward $v_i$.
\begin{itemize}
\item[$a_1$:] Every element $v_j$ of $D_i^+$ cannot be in $\mathcal P^1_i$, nor can vertices $v_k$ for which their is a directed path from $v_j$ to $v_k$ be in $\mathcal P^1_i$.  Hence, there can be at most $a_1(v_j)$ vertices of $V_j$ in $\mathcal P^1_i$.  For every vertex $v_j$ in $D_i^-$, any maximal-size disconnected-path family of $V_j$ can be chosen.  Thus, there can be at most $a_6(v_j)$ vertices of $V_j$ in $\mathcal P^1_i$. Thus,
$$a_1(v_i)\le \sum_{\substack{(v_i,v_j)\in E(H)\\i\le j}}a_1(v_j)+\sum_{\substack{(v_j,v_i)\in E(H)\\i\le j}}a_6(v_j)$$
Because the union of these maximal-size disconnected-path families is a disconnected-path family of $V_i$, we have that equality holds.
\item[$a_2$:] This is analogous to the proof for $a_1$, except the directions of the edges are reversed.
\item[$a_3$:] There are two cases to consider, depending on whether the path through $v_i$ in $\mathcal P^3_i$ contains any vertices besides $v_i$.  If the path does have additional vertices, let $v_j^+$ be the element of $D_i^+$ through which the path from $v_i$ in $\mathcal P^3_i$ continues. At most $a_3(v_{j^+})$ vertices of $\mathcal P^3_i$ can be from the tree $V_{j^+}$.  For any other vertex $v_{k^+}$ in $D_i^+$, at most $a_1(v_{k^+})$ vertices of $\mathcal P^3_i$ can be from the tree $V_{k^+}$.  If the path through $v_i$ does not encompass any other vertices, then each subtree $V_{k^+}$ in $D_i^+$ contributes at most $a_1(v_{k^+})$ vertices to $\mathcal P^3_i$  Each vertex of $v_{j^-}$ in $D_i^-$ can contribute up to $a_2(v_{j^-})$ vertices to $\mathcal P^3_i$.  Because there may be choices for which vertex the path from $v_i$ passes through, we must take the maximum of all the possibilities
\begin{align*}a_3(v_i)\le{}&1+\sum_{\substack{(v_i,v_j)\in E(H)\\i\le j}}a_1(v_j)+\sum_{\substack{(v_j,v_i)\in E(H)\\i\le j}}a_2(v_j)\\&+ \max\left(0,\max_{\substack{(v_i,v_j)\in E(H)\\i\le j}}(a_3(v_j)-a_1(v_j))\right)\end{align*}
As the vertices of these maximal-size disconnected-path families form a disconnected-path family of $V_i$, equality is attainable.
\item[$a_4$:] This is analogous to the proof for $a_3$, except the directions of the edges are reversed.
\item[$a_5$:] If the path through $v_i$ in $\mathcal P^5_i$ has $v_i$ itself as one of its endpoints, then we can use $\max(a_3(v_i), a_4(v_i))$ as an upper bound.  Thus, let us assume the contrary.  Let $v_{j^+}$ be the element of $D_i^+$ through which the path from $v_i$ through $\mathcal P^5_i$ continues. At most $a_3(v_{j^+})$ vertices of $\mathcal P^5_i$ can be from the tree $V_{j^+}$.  For any other vertex $v_{k^+}$ in $D_i^+$, at most $a_1(v_{k^+})$ vertices of $\mathcal P^5_i$ can be from the tree $V_{k^+}$.  Let $v_{j^-}$ be the element of $D_i^-$ from which the path in $\mathcal P^5_i$ through $v_i$ enters $v_i$. At most $a_4(v_{j^-})$ vertices of $\mathcal P^5_i$ can be from the tree $V_{j^-}$.  For any other vertex $v_{k^-}$ in $D_i^-$, at most $a_2(v_{k^-})$ vertices of $\mathcal P^5_i$ can be from the tree $V_{k^-}$. Because there may be multiple elements in $D^+_i$ or $D^-_i$, we must take the maximum of all possibilities.  These facts can be combined to yield the inequality
\begin{align*}a_5(v_i)={}&1+\sum_{\substack{(v_i,v_j)\in E(H)\\i\le j}}a_1(v_j)+\sum_{\substack{(v_j,v_i)\in E(H)\\i\le j}}a_2(v_j)\\ &+ \max\left(0,\max_{\substack{(v_i,v_j)\in E(H)\\i\le j}}(a_3(v_j)-a_1(v_j))\right)\\ &+ \max\left(0,\max_{\substack{(v_j,v_i)\in E(H)\\i\le j}}(a_4(v_j)-a_2(v_j))\right)\end{align*}
As from the previous cases, equality is attainable because the vertices of these maximal-size disconnected-path families form a disconnected-path family of $V_i$.
\item[$a_6$:] In the tree $V_i$, either there is a path through $v_i$ in $\mathcal P^i_6$, which is accounted for in $a_5$, there is no path through $v_i$ and no descendants of $v_i$ are in $\mathcal P^6_i$, which is accounted for in $a_1$, or there is no path through $v_i$ and no ancestors of $v_i$ are in $\mathcal P^6_i$, which is accounted for in $a_2$.
\end{itemize}

\end{proof}
Utilizing this recursion, we can construct a linear-time algorithm for computing $p(H)$.  First, it takes constant time to compute the values of the six functions for $v_n$.  Then, given the values of these six functions for the vertices $v_j,\hdots, v_n$, we can compute these values for $v_{j-1}$ in $O(d(v_{j-1}))$ time.  Thus, we can compute $a_6(v_1)$ in
$$O\left(\sum_{i=1}^nd(v_i)\right)=O(|E(H)|)=O(n)$$
time, as desired.
\end{proof}

\end{document}